\documentclass[12pt]{article}
\usepackage[francais,english]{babel}
\usepackage[T1]{fontenc}
\usepackage{kpfonts}
\usepackage{color}
\usepackage{amssymb}
\usepackage{geometry}
\RequirePackage{filecontents}
\usepackage{amsmath,mathtools}
\usepackage{graphicx}
\usepackage{sectsty}
\usepackage{lipsum}
\usepackage{mathrsfs}
\usepackage{amsthm}
\allsectionsfont{\centering}
\newtheorem{theorem}{Theorem}[section]
\newtheorem{corollary}{Corollary}[section]

\newtheorem{proposition}{Proposition}[section]
\newtheorem{remark}{Remark}[section]
\newtheorem{lemma}{Lemma}[section]

\numberwithin{equation}{section}
\begin{document}
\large
\begin{center}
   \textbf{\Large{Bargmann-Dirichlet Spaces from  Magnetic  Laplacians  \\ and theirs Bargmann Transforms}}
\end{center}
\begin{center}
   \textbf{Nour eddine Askour $^{1,3,a}$, \textbf{Adil Belhaj} $^{2,b}$ and  Mohamed Bouaouid $^{1,c}$}
\end{center}
\begin{center}
$^{1}$ Department of Mathematics, Sultan Moulay Slimane University, Faculty
        of Sciences and Technics, Béni Mellal, Morocco.
\end{center}
\begin{center}
$^{2}$ EPTM, Department of Physics, Sultan Moulay Slimane University, Polydisciplinary Faculty, Béni Mellal, Morocco.
\end{center}
\begin{center}
$^{3}$ Department of Mathematics, Mohammed V University, Faculty of Sciences, Rabat, Morocco.
\end{center}
\begin{center}
$^{a}$ n.askour@usms.ma, \hspace{0.25cm}$^{b}$ belhaj@unizar.es\hspace{0.25cm} and \hspace{0.25cm}$^{c}$ bouaouidfst@gmail.com
\end{center}
\section*{Abstract}
We reconsider the Bargmann-Dirichlet space on the complex plane
$\mathbb{C}$ and its generalizations  considered in \cite{Elh}.
Concretely, we first present a new characterization of such spaces
as harmonic spaces of the magnetic Laplacian with suitable domains.
Then, we elaborate an associated unitary integral transforms of
Bargmann type.
\begin{description}
  \item[Keywords:] Hilbert spaces with reproducing kernels; Partial differential operators; Bargmann-Dirichlet spaces;
   Bargmann transforms; Magnetic Laplacians; Hermite polynomials.

  \item[2010 Mathematics Subject Classification:] 74S70; 47Axx; 81Q10; 35Qxx.
\end{description}
\newpage
\tableofcontents
\newpage
\section{Introduction and summarized results}
The Bargmann transform is an integral transformation which intertwines the Sch\"{o}dinger representation and the complex wave
 (Fock-Bargmann) representation of the quantum harmonic oscillator. It plays an important role in the complex analysis and the
  phase space formulation of quantum mechanics. In the case where this harmonic oscillator corresponds to a charged particle
   moving in the plane $\mathbb{R}^{2}$ under a orthogonal magnetic field $\overrightarrow{B}$ of intensity $\nu=1$, the Bargmann transform
\begin{align}\label{E1.1}
 \nonumber B\hspace{0.2cm}:L^{2}(\mathbb{R},& \hspace{0.2cm}dx)\longrightarrow \mathcal{A}^{2,1}(\mathbb{C})\subset L^{2}(\mathbb{C},
 \hspace{0.2cm}e^{-\mid z\mid^{2}}d\lambda(z))\\
 &\varphi\longmapsto B[\varphi](z):=\pi^{-\frac{3}{4}}\int_{-\infty}^{+\infty}\exp(-\frac{x^{2}}{2}+\sqrt{2}xz-\frac{z^{2}}{2})\varphi(x)dx
 \end{align}
maps isometrically the space $L^{2}(\mathbb{R},\hspace{0.2cm}dx)$ on the Bargmann-Fock space $\mathcal{A}^{2,1}(\mathbb{C})$ of the
 holomorphic functions integrable with respect to the Gaussian measure $e^{-\mid z\mid^{2}}d\lambda(z)$. Here, $d\lambda(z)$
 is the ordinary area measure. The inner product in $\mathcal{A}^{2,1}(\mathbb{C})$ is inherited from
 $L^{2}(\mathbb{C},\hspace{0.2cm}e^{-\mid z\mid^{2}}d\lambda(z))$.\\
The Bargmann-Fock space is a convenient setting for many problems in complex analysis, physics and engineering  applications
 \cite{Fil}. In \cite{Ask}, the authors have linked the Bargmann-Fock space to the context of theoretical physics by proving that this space corresponds to states associated with a minimal energy. The schr\"{o}dinger operator with uniform constant magnetic field
  $\overrightarrow{B}$ of intensity $\nu>0$ and perpendicular to the plane $\mathbb{R}^{2}$ can be written as
\begin{align}\label{E1.2}
H_{\nu}=-\frac{1}{4}\{(\frac{\partial}{\partial x}+i\nu y)^{2}+(\frac{\partial}{\partial y}-i\nu x)^{2}\}
\end{align}
acting on the Hilbert space $L^{2}(\mathbb{R}^{2},\hspace{0.2cm}dxdy)$. The spectral properties of this operator have been
 investigated extensively in many works \cite{Ask,Avr,Hel}. Further, a more general class of such an operator in higher dimensional cases was discussed in \cite{Mal1,Mal2,Shi} with connection to the stochastic oscillator integral.\\
The Hamiltonian operator $H_{\nu}$, defined in $(\ref{E1.2})$, can be intertwined with the following magnetic Laplacian
\begin{align}\label{E1.3}
\tilde{\Delta}_{\nu}=-\frac{\partial^{2}}{\partial z\partial\overline{z}}+\nu\overline{z}\frac{\partial}{\partial\overline{z}}, \hspace{0.25cm}\nu>0,
\end{align}
acting on the Hilbert space $L^{2}(\mathbb{C},\hspace{0.2cm}e^{-\nu\mid z\mid^{2}}d\lambda(z))$ via the connection formula
\begin{align}\label{E1.4}
T_{\nu}\circ (H_{\nu}-\frac{\nu}{2})\circ T_{\nu}^{-1}=\tilde{\Delta}_{\nu },
\end{align}
where $T_{\nu}$ is the following unitary isomorphism
\begin{align}\label{E1.5}
\nonumber T_{\nu}:\hspace{0.25cm}L^{2}(\mathbb{R},&\hspace{0.2cm}dx)\longrightarrow L^{2}(\mathbb{C},\hspace{0.2cm}e^{-\nu\mid z\mid^{2}}d\lambda(z))\\
        &\varphi\longmapsto T_{\nu}[\varphi](z)=e^{\frac{\nu\mid z\mid^{2}}{2}}\varphi.
\end{align}
The formula $(\ref{E1.4})$ connects in an equivalent way the spectral theory of the two operators $H_{\nu}$ and $\tilde{\Delta}_{\nu }$.\\
From \cite{Ask}, the spectrum of the self-adjoint hamiltonian $H_{\nu}$ is reduced to eigenvalues and it is given by
\begin{align}\label{E1.6}
\sigma(H_{\nu})=\{\frac{\nu}{2},\frac{\nu}{2}+1,...\}.
\end{align}
The eigenspace associated with the lower energy $\frac{\nu}{2}$ (the
smallest eigenvalue) in the Schr\"{o}dinger representation
 $H_{\nu}\varphi=\frac{\nu}{2}\varphi$ corresponds to the eigenstate of the magnetic Laplacian (Laundau states) $\tilde{\Delta}_{\nu }$
  associated with the eigenvalue $0$ in the Bargmann representation
\begin{align}\label{E1.7}
\tilde{\Delta}_{\nu }\psi=0,\hspace{0.25cm}\psi\in L^{2}(\mathbb{C},\hspace{0.2cm}e^{-\nu\mid z\mid^{2}}d\lambda(z)).
\end{align}
In \cite{Ask2}, the authors proved that, for the intensity $\nu=1$, the harmonic space
\begin{align}\label{E1.8}
\mathcal{E}_{0}^{2,1}=\{\psi\in L^{2}(\mathbb{C},\hspace{0.2cm}e^{-\mid z\mid^{2}}d\lambda(z)),\hspace{0.2cm}\tilde{\Delta}_{1}\psi=0\}
\end{align}
coincides with the Bargmann-Fock space defined by
\begin{align}\label{E1.9}
\mathcal{A}^{2,1}(\mathbb{C}):=\{f:\mathbb{C}\longrightarrow\mathbb{C}, \hspace{0.15cm} \mbox{holomorphic and} \hspace{0.15cm}
\int_{\mathbb{C}}\mid f(z)\mid^{2}e^{-\mid z\mid^{2}}d\lambda(z)<+\infty\}.
\end{align}
In other words, the Bargmann transform maps, in an isometrical way, the Hilbert space $L^{2}(\mathbb{R},\hspace{0.15cm}dx)$ to
 the harmonic space $\mathcal{E}_{0}^{2,1}$ of the magnetic Laplace $\tilde{\Delta}_{1}$.\\
For $\nu>0$, the present work   will be devoted to an analogue case for the following  Bargmann-Dirichlet space  on the complex plane
\begin{align}\label{E1.10}
\mathcal{A}_{1}^{2,\nu}(\mathbb{C})=:\{f:\mathbb{C}\longrightarrow\mathbb{C},
\hspace{0.15cm} \mbox{holomorphic and}
\hspace{0.15cm}\int_{\mathbb{C}} \mid f^{'}(z)\mid^{2}e^{-\nu\mid
z\mid^{2}}d\lambda(z)<+\infty\},
\end{align}
where $f^{'}$ is the complex derivative of the function $f$. This space has been considered in \cite{Elh} as a Hilbert space equipped with the following scalar product
\begin{align}\label{E1.11}
<f,g>_{\nu,1}=\frac{\pi}{\nu}f(0)\overline{g(0)}+<f^{'},g^{'}>_{\nu},
\end{align}
where $<,>_{\nu}$ is the scalar product associated with the norm
\begin{align}\label{E1.12}
\parallel f\parallel^{2}_{\nu}=\int_{\mathbb{C}}\mid f(z)\mid^{2}e^{-\nu\mid z\mid^{2}}d\lambda(z),
\end{align}
defined on the Hilbert space
$L^{2}(\mathbb{C},\hspace{0.2cm}e^{-\nu\mid z\mid^{2}}d\lambda(z))$.
The authors of  the work  \cite{Elh} have showed that the Bargmann-Dirichlet space $(\mathcal{A}_{1}^{2,\nu}(\mathbb{C}),<,>_{\nu,1})$ is a reproducing kernel Hilbert space (R.K.H.S) and its reproducing kernel is given by
\begin{align}\label{E1.13}
\tilde{K}(z,w)=\frac{\pi}{\nu}\{1+z\overline{w}\prescript{}{2}{F}_2^{}(1,1;2,2;\nu z \overline{w})\},
\end{align}
where $\prescript{}{2}{F}_2^{}$ is the hypergeometric function defined by \cite[p.62]{Mag}
\begin{align}\label{E1.14}
\prescript{}{2}{F}_2^{}(a,b;a^{'},b^{'};x)=\sum_{k=0}^{+\infty}\frac{(a)_{k}(b)_{k}}{(a^{'})_{k}(b^{'})_{k}}\frac{x^{k}}{k!}.
\end{align}
We first prove that the Bargmann-Dirichlet space $\mathcal{A}_{1}^{2,\nu}(\mathbb{C})$ is the null space of the
magnetic Laplacian $\Delta_{\nu}$ with a suitable domain. Precisely, we have the following proposition.
\begin{proposition}\label{P2.4}
Let $\nu>0$ and $\Delta_{\nu}$ be the partial differential operator defined by
\begin{align}\label{E1.15}
\Delta_{\nu}=-\frac{\partial^{2}}{\partial z\partial\overline{z}}+\nu\overline{z}\frac{\partial}{\partial\overline{z}},
\end{align}
acting on the Hilbert space  $L^{2,\nu}(\mathbb{C}):=L^{2}(\mathbb{C},\hspace{0.2cm}e^{-\nu\mid z\mid^{2}}d\lambda(z))$ with the dense
 domain \begin{align}\label{E1.16}
D(\Delta_{\nu})=\{F\in L^{2,\nu}(\mathbb{C}),\hspace{0.25cm}\Delta_{\nu}F\in L^{2,\nu}(\mathbb{C})\hspace{0.25cm}\mbox{and}\hspace{0.25cm}
\frac{\partial F}{\partial z}\in L^{2,\nu}(\mathbb{C})\}.
\end{align}
Then, the Bargmann-Dirichlet space $\mathcal{A}^{2,\nu}_{1}(\mathbb{C})$ is the null space of $\Delta_{\nu}$, that is,
\begin{align}\label{E1.17}
\mathcal{A}^{2,\nu}_{1}(\mathbb{C})=\{F\in D(\Delta_{\nu}),\hspace{0.25cm}\Delta_{\nu}F=0\}.
\end{align}
\end{proposition}
Note that the operator $\Delta_{\nu}$ defined by $(\ref{E1.15})$ and $(\ref{E1.16})$ is a non self-adjoint symmetric operator.
 This result connect the Bargmann-Dirichlet space to the context of quantum dynamics via a non self-adjoint operator. This fact goes in the sprit of the new vision on the role of the non self-adjoint operator in quantum physics. For more details about this subject, we refer to \cite{Bag} and the references therein. Also, as for the classical Bargmann-Fock space, we elaborate a unitary isomorphism from $L^{2}(\mathbb{R},\hspace{0.2cm}dx)$ into the Bargmann-Dirichlet space $\mathcal{A}^{2,\nu}_{1}(\mathbb{C})$. This isomorphism will be called, in the present work, the Bargmann-Dirichlet transform. More precisely, we have the following result.
\begin{theorem}\label{T}
For the Bargmann-Dirichlet space $\mathcal{A}_{1}^{2,\nu}(\mathbb{C})$, we have the following associated Bargmann-Dirichlet transform
\begin{align}\label{E1.18}
 \nonumber B_{\nu}\hspace{0.2cm}:L^{2}(\mathbb{R},& \hspace{0.2cm}dx)\longrightarrow \mathcal{A}_{1}^{2,\nu}(\mathbb{C})\\
 &\varphi\longmapsto B_{\nu}[\varphi](z):=\int_{-\infty}^{+\infty}K_{\nu}(z,x)\varphi(x)dx,
 \end{align}
  where the integral kernel $K_{\nu}(z,x)$ is given by
 \begin{align}\label{E1.19}
 K_{\nu}(z,x)=\sqrt{\nu}(\pi)^{\frac{-3}{4}}e^{-\frac{x^{2}}{2}}[1+\frac{\sqrt{2}}{\sqrt{\pi}}z\int_{0}^{+\infty}\sqrt{t}e^{-t}
\exp(x\sqrt{2\nu}e^{-t}z-\nu e^{-2t}\frac{z^{2}}{2})H_{1}(x-\sqrt{\frac{\nu}{2}}e^{-t}z)dt],
\end{align}
with $H_{1}(x)=2x$ denotes the second classical Hermite polynomial
\cite[p.250]{Mag}.
\end{theorem}
In the analogue way of the Bargmann-Dirichlet space $\mathcal{A}_{1}^{2,\nu}(\mathbb{C})$, we consider the following generalized
Bargmann-Dirichlet space $\mathcal{A}_{m}^{2,\nu}(\mathbb{C})$, $\nu>0$, $m\geq2$, defined by
\begin{align}\label{E1.20}
\mathcal{A}_{m}^{2,\nu}(\mathbb{C}):=\{f:\mathbb{C}\longrightarrow\mathbb{C},
\hspace{0.15cm} \mbox{holomorphic and}
\hspace{0.15cm}\int_{\mathbb{C}}\mid f^{(m)}(z)\mid^{2}e^{-\nu\mid
z\mid^{2}}d\lambda(z)<+\infty\},
\end{align}
where $f^{(m)}(z):=\frac{\partial^{m} f(z)}{\partial z^{m}}$ indicates the complex $m-$derivative of the function $f$.\\
For fixed non-negative integers $m=2,3,...$, any holomorphic function  $f(z)=\displaystyle{\sum_{k=0}^{+\infty}}a_{k}z^{k}$ on
$\mathbb{C}$ can be written as
\begin{align}\label{E1.21}
f(z)=f_{1,m}(z)+f_{2,m}(z),
\end{align}
where  $f_{1,m}(z)=\displaystyle{\sum_{k=0}^{m-1}}a_{k}z^{k}$ and
$f_{2,m}(z)=\displaystyle{\sum_{k=m}^{+\infty}}a_{k}z^{k}$.
 The space $\mathcal{A}_{m}^{2,\nu}(\mathbb{C})$ can be equipped with the following norm
\begin{align}\label{E1.22}
\parallel f\parallel_{\nu,m}^{2}=\parallel f_{1,m}\parallel_{\nu}^{2}+\parallel f_{2,m}^{(m)}\parallel_{\nu}^{2},
\end{align}
where $\parallel.\parallel_{\nu}$ is the norm defined on the Hilbert space $L^{2}(\mathbb{C},\hspace{0.2cm}e^{-\nu\mid z\mid^{2}}d\lambda(z))$
 by the formula $(\ref{E1.12})$.\\
It is noted that the hermitian inner product $<,>_{\nu,m}$ associated with the norm\\
 $\parallel. \parallel_{\nu,m}$ is given through
\begin{align}\label{E1.23}
<f,g>_{\nu,m}=<f_{1,m},g_{1,m}>_{\nu}+<f_{2,m}^{(m)},g_{2,m}^{(m)}>_{\nu}.
\end{align}
The reproducing kernel of the generalized Bargmann-Dirichlet space $\mathcal{A}_{m}^{2,\nu}(\mathbb{C})$ is given by the following expression \cite{Elh}
\begin{align}\label{E1.24}
\tilde{K}_{m,\nu}(z,w)=\frac{\pi}{\nu}\{\sum_{j=0}^{m-1}\frac{(\nu z \overline{w})^{j}}{j!}+\frac{(z\overline{w})^{m}}{(m!)^{2}}\prescript{}{2}{F}_2^{}(1,1;m+1,m+1;\nu z \overline{w})\},
\end{align}
where $\prescript{}{2}{F}_2^{}$ is the hypergeometric function defined in $(\ref{E1.14})$.\\
The generalized Bargmann-Dirichlet space $\mathcal{A}_{m}^{2,\nu}(\mathbb{C})$ could be  connected with the
context of quantum mechanics of a planar particle moving under the action of a uniform constant perpendicular magnetic field of intensity $\nu$. Precisely, we have the following proposition
\begin{proposition}\label{P3.1}
Let  $\nu>0$, $m\in \mathbb{Z}_{+}$, $m\geq2$  and $\Delta_{\nu}$ be the partial differential operator which reads as
\begin{align}\label{E1.25}
\Delta_{\nu}=-\frac{\partial^{2}}{\partial z\partial\overline{z}}+\nu\overline{z}\frac{\partial}{\partial\overline{z}},
\end{align}
acting on the Hilbert space  $L^{2,\nu}(\mathbb{C}):=L^{2}(\mathbb{C},\hspace{0.2cm}e^{-\nu\mid z\mid^{2}}d\lambda(z))$ with
 the dense domain \begin{align}\label{E1.26}
D_{m}(\Delta_{\nu})=\{F\in L^{2,\nu}(\mathbb{C}),\hspace{0.25cm}\Delta_{\nu}F\in L^{2,\nu}(\mathbb{C})\hspace{0.25cm}\mbox{and}\hspace{0.25cm}\frac{\partial^{m} F}{\partial z^{m} }\in L^{2,\nu}(\mathbb{C})\}.
\end{align}
Then, we have the following characterization
\begin{align}\label{E1.27}
\mathcal{A}^{2,\nu}_{m}(\mathbb{C})=\{F\in D_{m}(\Delta_{\nu}),\hspace{0.25cm}\Delta_{\nu}F=0\}.
\end{align}
\end{proposition}
For this generalized Bargmann-Dirichlet space $\mathcal{A}_{m}^{2,\nu}(\mathbb{C})$, we establish the following generalized
 Bargmann-Dirichlet transform. Precisely, we have the following result.

\begin{theorem}\label{T3.1} Let $\nu >0$ and $m\in \mathbb{Z}_{+}$, $m\geq2$. Then, we have the following unitary isomorphism
\begin{align}\label{E1.28}
 \nonumber B_{\nu,m}\hspace{0.2cm}:L^{2}(\mathbb{R},& \hspace{0.2cm}dx)\longrightarrow \mathcal{A}_{m}^{2,\nu}(\mathbb{C})\\
 &\varphi\longmapsto B_{\nu,m}[\varphi](z):=\int_{-\infty}^{+\infty}K_{\nu,m}(z,x)\varphi(x)dx.
 \end{align}
 The integral kernel $K_{\nu,m}(z,x)$ takes the form
 \begin{align}\label{E1.29}
\nonumber K_{\nu,m}(z,x)&=\sqrt{\nu}(\pi)^{\frac{-3}{4}}e^{\frac{-x^{2}}{2}}[\sum_{k=0}^{m-1}\bigg(\sqrt{\frac{\nu}{2}}z\bigg)^{k}\frac{H_{k}(x)}{k!}\\
&+\bigg(\sqrt{\frac{2}{\pi}}\bigg)^{m}z^{m}\int_{0}^{+\infty}\varpi_{m}(t)
\exp(x\sqrt{2\nu}e^{-t}z-\nu e^{-2t}\frac{z^{2}}{2})H_{m}(x-\sqrt{\frac{\nu}{2}}e^{-t}z)dt],
\end{align}
where $H_{m}(x)$ is the classical Hermite polynomial and
$\varpi_{m}(t)$ is the function defined by
\begin{align}\label{E1.30}
\varpi_{m}(t)=(\sqrt{t}e^{-t})*(\sqrt{t}e^{-2t})*...*(\sqrt{t}e^{-mt}), \hspace{0.25cm}m\geq2.
\end{align}
The notation $f*g$ means the convolution product \cite[p.91]{Sch}
\begin{align}\label{E1.31}
f*g(x)=\int_{0}^{x}f(x-y)g(y)dy.
\end{align}
\end{theorem}
For $m=2$, we can explicit the function $\varpi_{m}(t)$. Then, we have the following more precise result.
\begin{proposition}\label{P3.12} Let $\nu>0$ and $m=2$. Then, we have the following unitary isomorphism
\begin{align}\label{E1.32}
 \nonumber B_{\nu,2}\hspace{0.2cm}:L^{2}(\mathbb{R},& \hspace{0.2cm}dx)\longrightarrow \mathcal{A}_{2}^{2,\nu}(\mathbb{C})\\
 &\varphi\longmapsto B_{\nu,2}[\varphi](z):=\int_{-\infty}^{+\infty}K_{\nu,2}(z,x)\varphi(x)dx,
 \end{align}
  where the integral kernel $K_{\nu,2}(z,x)$ is given by
 \begin{align}\label{E1.33}
\nonumber K_{\nu,2}(z,x)&=\sqrt{\nu}(\pi)^{\frac{-3}{4}}e^{\frac{-x^{2}}{2}}[1+\sqrt{2\nu}xz\\
&+\frac{z^{2}}{4}\int_{0}^{+\infty}t^{2}
\exp(-2t+x\sqrt{2\nu}e^{-t}z-\nu e^{-2t}\frac{z^{2}}{2})\prescript{}{1}{F}_1^{}(\frac{3}{2};3;t)H_{2}(x-\sqrt{\frac{\nu}{2}}e^{-t}z)dt],
\end{align}
with $\prescript{}{1}{F}_1^{}(\alpha;\beta;t)=\displaystyle{\sum_{j=0}^{+\infty}}\frac{(\alpha)_{j}}{(\beta)_{j}}\frac{z^{j}}{j!}$
is the confluent hypergeometric function and $H_{2}(y)=4y^{2}-2$ is the third Hermite polynomial.
\end{proposition}
The paper is organized as follows. In section 2, we review some results on the  spectral analysis of the Schr\"{o}dinger operator with a uniform constant magnetic field on the complex plane $\mathbb{C}$. Then, we characterize the Barmann-Dirichlet space as a harmonic space of a magnetic Laplacian with a suitable domain. We end this section by establishing an integral transformation that realizes an isometry between $L^{2}(\mathbb{R},\hspace{0.2cm}dx)$ and the Bargmann-Dirichlet space $\mathcal{A}_{1}^{2,\nu}(\mathbb{C})$. Section 3 will be devoted to the generalized Bargmann-Dirichlet space $\mathcal{A}_{m}^{2,\nu}(\mathbb{C})$, $m\geq 2$. In this section, we follow the same lines as in the section 2. The last section is  reserved to concluding remarks.
\section{Bargmann-Dirichlet space}
\subsection{Bargmann-Dirichlet space from magnetic Laplacian}
Let $\nu>0$ and $L^{2,\nu}(\mathbb{C}):=L^{2}(\mathbb{C},\hspace{0.2cm}e^{-\nu\mid z\mid^{2}}d\lambda(z))$ is the Hilbert space of all square-integrable functions on $\mathbb{C}$ with respect to the Gaussian measure $d\mu_{\nu}(z):=e^{-\nu\mid z\mid^{2}}d\lambda(z)$,
  where $d\lambda(z)$ is the Lebesgue measure. The hermitian inner product is defined by
\begin{align}\label{E2.1}
<f,g>_{\nu}=\int_{\mathbb{C}}f(z)\overline{g(z)}e^{-\nu\mid z\mid^{2}}d\lambda(z),
\end{align}
and the associated norm is given by
\begin{align}\label{E2.2}
\parallel f\parallel^{2}_{\nu}=\int_{\mathbb{C}}\mid f(z)\mid^{2}e^{-\nu\mid z\mid^{2}}d\lambda(z).
\end{align}
In the analogue way of the weighted Bergmann-Dirichlet space on the unit disk \cite{Elf}, we consider the following functional space, called the
 Bargmann-Dirichlet space on $\mathbb{C}$ (\cite{Elh}). Here, this space will be denoted  $\mathcal{A}_{1}^{2,\nu}(\mathbb{C})$ and defined by
\begin{align}\label{E2.3}
\mathcal{A}_{1}^{2,\nu}(\mathbb{C}):=\{f:\mathbb{C}\longrightarrow\mathbb{C}, \hspace{0.15cm} \mbox{holomorphic and} \hspace{0.15cm}
\int_{\mathbb{C}}\mid f^{'}(z)\mid^{2}e^{-\nu\mid z\mid^{2}}d\lambda(z)<+\infty\},
\end{align}
where $f^{'}(z)$ is the complex derivative defined by $f^{'}(z):=\frac{\partial f(z)}{\partial z}=\frac{1}{2}(\frac{\partial f(z)}{\partial x}-i\frac{\partial f(z)}{\partial y}),$ with $z=x+iy.$\\
On the Bargmann-Dirichlet space
$\mathcal{A}_{1}^{2,\nu}(\mathbb{C})$, we take  the following norm
 \begin{align}\label{E2.4}
 \parallel f\parallel^{2}_{\nu,1}=\frac{\pi}{\nu}\mid f(0)\mid^{2}+\parallel f^{'}\parallel^{2}_{\nu},
 \hspace{0.25cm} f\in \mathcal{A}_{1}^{2,\nu}(\mathbb{C}).
 \end{align}
Notice that the hermitian inner product $<,>_{\nu,1}$, associated with the norm $\parallel.\parallel_{\nu,1}$, is given  through
\begin{align}\label{E2.5}
<f,g>_{\nu,1}=\frac{\pi}{\nu}f(0)\overline{g(0)}+<f^{'},g^{'}>_{\nu}.
\end{align}
It is recalled that the monomials $e_{j}(z)=z^{j}$ belong to $\mathcal{A}_{1}^{2,\nu}(\mathbb{C})$ (\cite{Elh}).
 They are pairwise orthogonal with respect to the hermitian scalar product $<,>_{\nu,1}$ with
\begin{align}\label{E2.6}
\parallel e_{j}\parallel^{2}_{\nu,1}=
\left\{
  \begin{array}{ll}
    \frac{\pi}{\nu}, &\mbox{for} \hspace{0.25cm}j=0, \hbox{} \\\\
    \frac{\pi(j!)^{2}}{\nu^{j}\Gamma(j)}, &\mbox{for} \hspace{0.25cm}j\geq1, \hbox{}
  \end{array}
\right.
\end{align}
where $\Gamma$ is the classical Euler gamma function
\cite[p.1]{Mag}. Then, an orthonormal basis for the space
$\mathcal{A}_{1}^{2,\nu}(\mathbb{C})$ can be written as
\begin{align}\label{E2.7}
   \psi_{j}(z)=\left\{
                 \begin{array}{ll}
                   \sqrt{\frac{\nu}{\pi}},&j=0,\hbox{} \\\\
                   \frac{(\nu)^{\frac{j}{2}}}{\sqrt{\pi j(j!)}}z^{j},&j\geq 1. \hbox{}
                 \end{array}
               \right.
\end{align}
It is well know that the Bargmann-Dirichlet space $\mathcal{A}_{1}^{2,\nu}(\mathbb{C})$ equipped with the above hermitian scalar product is a reproducing kernel
Hilbert space \cite{Elh}. The reproducing kernel is given by the Papadakis formula \cite{Pau}
\begin{align}\label{E2.8}
K(z,w)=\sum_{k=0}^{+\infty}\psi_{j}(z)\overline{\psi_{j}(w)}.
\end{align}
According to \cite{Elh}, this kernel has been given explicitly by the following formula
\begin{align}\label{E2.9}
K(z,w)=\frac{\pi}{\nu}\{1+z\overline{w}\prescript{}{2}{F}_2^{}(1,1;2,2;\nu z \overline{w})\},
\end{align}
where $\prescript{}{2}{F}_2^{}$ is the hypergeometric function defined by \cite[p.62]{Mag}
\begin{align}\label{E2.10}
\prescript{}{2}{F}_2^{}(a,b;a^{'},b^{'};x)=\sum_{k=0}^{+\infty}\frac{(a)_{k}(b)_{k}}{(a^{'})_{k}(b^{'})_{k}}\frac{x^{k}}{k!}.
\end{align}
It is noted that we have the following membership test \cite{Elh}.
\begin{align}\label{E2.11}
f(z):=\sum_{j=0}^{+\infty} a_{j}z^{j}\in \mathcal{A}_{1}^{2,\nu}(\mathbb{C})\Longleftrightarrow
\sum_{j=0}^{+\infty}\frac{j(j!)}{\nu^{j}} \mid a_{j} \mid^{2}<+\infty.
\end{align}
In order to relate the Bargmann-Dirichlet space $\mathcal{A}_{1}^{2,\nu}(\mathbb{C})$ with the Landau states on
 the complex plane, we need to recall the following magnetic Laplacian representing Shr\"{o}dinger operator with a magnetic field having the magnitude $\nu>0$
\begin{align}\label{E2.12}
\tilde{\Delta}_{\nu}=-\frac{\partial^{2}}{\partial z\partial\overline{z}}+\nu\overline{z}\frac{\partial}{\partial\overline{z}},
\end{align}
acting on the Hilbert space
$L^{2,\nu}(\mathbb{C}):=L^{2}(\mathbb{C},\hspace{0.2cm}e^{-\nu\mid z\mid^{2}}d\lambda(z))$ with the maximal domain
\begin{align}\label{E2.13}
D(\tilde{\Delta}_{\nu})=\{F\in L^{2,\nu}(\mathbb{C}),\hspace{0.25cm}\tilde{\Delta}_{\nu}F\in L^{2,\nu}(\mathbb{C})\}.
\end{align}
Note that the operator $\tilde{\Delta}_{\nu}$ can be unitarily intertwined with the operator $\tilde{\Delta}_{1}$ $(\nu=1)$. Precisely, we have
\begin{align}\label{E2.14}
T_{\nu}\circ\tilde{\Delta}_{\nu}\circ T_{\nu}^{-1}=\nu \tilde{\Delta}_{1},
\end{align}
where $T_{\nu}$ is the following isometry defined by
\begin{align}\label{E2.15}
\nonumber T_{\nu}:\hspace{0.25cm}L^{2}(\mathbb{C},&\hspace{0.2cm}e^{-\nu\mid z\mid^{2}}d\lambda(z))\longrightarrow L^{2}(\mathbb{C},\hspace{0.2cm}e^{-\mid z\mid^{2}}d\lambda(z))\\
        &\varphi\longmapsto T_{\nu}[\varphi](z)=\frac{1}{\nu}\varphi(\frac{z}{\sqrt{\nu}}).
\end{align}
From spectral theory point of view, the spectral tools of the operator $\tilde{\Delta}_{\nu}$ can be obtained from
 $\tilde{\Delta}_{1}$ by using the relation $(\ref{E2.14})$.\\
The $n-$dimensional analogue of the operator  $\tilde{\Delta}_{1}$ is given by
\begin{align}\label{E2.16}
\tilde{\Delta}=-\sum_{j=1}^{n}\frac{\partial^{2}}{\partial z_{j}\partial\overline{z_{j}}}+\sum_{j=1}^{n}\overline{z_{j}}\frac{\partial}{\partial\overline{z_{j}}},
\end{align}
which has been extensively studied in \cite{Ask}. By exploiting the results given in \cite{Ask} for the case $n=1$,
 it is not hard to state for the operator $\tilde{\Delta}_{\nu}$ the following proposition.
\begin{proposition}\label{P2.1}
Let $\nu>0$ and let $\tilde{\Delta}_{\nu}$ be the operator defined by $(\ref{E2.12})$ and $(\ref{E2.13})$. Then, we have
 the following statements
\begin{enumerate}
  \item The operator $\tilde{\Delta}_{\nu}$ is unbounded self-adjoint operator on its maximal domain $D(\tilde{\Delta}_{\nu})$.
  \item The spectrum $\sigma(\tilde{\Delta}_{\nu})$ of the operator $\tilde{\Delta}_{\nu}$ is given by
  \begin{align}\label{E2.17}
  \sigma(\tilde{\Delta}_{\nu})=\{\nu\ell,\hspace{0.25cm}\ell=0,1,2,...\}.
  \end{align}
  \item For each $\ell\in\mathbb{Z}_{+}$, the spectral value $\nu\ell$ is an eigenvalue of the operator $\tilde{\Delta}_{\nu}$ and the corresponding
   eigenspace
\begin{align}\label{E2.18}
\mathcal{E}_{\ell}^{2,\nu}(\mathbb{C})=\{F\in
L^{2,\nu}(\mathbb{C}),\hspace{0.25cm}\tilde{\Delta}_{\nu}F=\nu\ell
F\}
\end{align}
is a R.K.H.S with the reproducing kernel
 \begin{align}\label{E2.19}
 \mathcal{K}^{\nu}_{\ell}(z,w)=\frac{\nu}{\pi}e^{\nu z \overline{w}}L_{\ell}(\nu\mid z-w\mid^{2}),
 \end{align}
where $L_{j}(x)$ is the classical Laguerre polynomial \cite[p.239]{Mag}.
\item The operator $\tilde{\Delta}_{\nu}$ admits the following spectral decomposition
\begin{align}\label{E2.20}
\tilde{\Delta}_{\nu}=\int_{-\infty}^{+\infty}\lambda dE_{\lambda}^{\nu},
\end{align}
where the spectral family $\{E_{\lambda}^{\nu},\hspace{0.2cm}\lambda \in \mathbb{R}\}$ is given by the following projectors
\begin{align}\label{E2.21}
\nonumber E_{\lambda}^{\nu}:\hspace{0.25cm}L^{2}(\mathbb{C},\hspace{0.2cm}e^{-\nu\mid z\mid^{2}}d\lambda(z)&)\longrightarrow
 L^{2}(\mathbb{C},\hspace{0.2cm}e^{-\nu\mid z\mid^{2}}d\lambda(z))\\
        &\varphi\longmapsto E_{\lambda}^{\nu}[\varphi],
\end{align}
where
\begin{align}\label{E2.22}
E_{\lambda}^{\nu}[\varphi](z)=
\left\{
  \begin{array}{ll}
    \frac{\pi}{\nu}\int_{\mathbb{C}}e^{\nu z \overline{w}}L_{[\lambda]}^{(1)}(\nu\mid z-w\mid^{2})\varphi(w)e^{-\nu\mid w\mid^{2}}d\lambda(w), &\lambda\geq0, \hbox{} \\\\
    0, &\lambda<0, \hbox{}
  \end{array}
\right.
\end{align}
with $[\lambda]$ is the integer part of the real number $\lambda$.
\end{enumerate}
\end{proposition}
\begin{remark}
In quantum dynamics, the operator $\tilde{\Delta}_{\nu}$ represents the Hamiltonian observable energy of a particle moving
 in the complex plane under the interaction of a uniform constant magnetic field with the magnitude $\nu$.
\end{remark}
The measure of this observable at a physical state $\psi\in L^{2,\nu}(\mathbb{C})$ $(\parallel \psi\parallel_{\nu}=1)$ is given
 by means of the spectral family $\{E_{\lambda}^{\nu}\}_{\lambda\in \mathbb{R}}$. Precisely, for two real values $\alpha<\beta$, the quantity
\begin{align}\label{E2.23}
\sigma^{\nu}_{(\alpha,\beta)}=<(E_{\beta}^{\nu}-E_{\alpha}^{\nu})\psi,\psi>_{\nu}
\end{align}
is the probability for which the measure of the observable energy $\tilde{\Delta}_{\nu}$ at the state $\psi$ to be in the interval $]\alpha,\beta[.$\\
Recall that for a densely defined self-adjoint operator $T$ acting
on an abstract complex Hilbert space $H$, the spectral decomposition
\cite[p.89]{Kon} can be written as
\begin{align}\label{E2.24}
T=\int_{-\infty}^{+\infty}\lambda dE_{\lambda},
\end{align}
where $\{E_{\lambda}\}_{\lambda\in \mathbb{R}}$ is the associated unique spectral family. The spectral density
\begin{align}\label{E2.25}
e_{\lambda}=\frac{dE_{\lambda}}{d\lambda}
\end{align}
is understood as an operator-valued distribution\cite{Est}. This is
an element of the space  $D^{'}(\mathbb{R},L(D(T),H))$,
where $L(D(T),H)$ is the space of bounded operator from the domain $D(T)$ of self-adjointness of $T$ to the whole Hilbert space $H$.\\
Furthermore, in the case where the Hilbert space $H$ is of $L^{2}-$type, that is, $H=L^{2}(M,\hspace{0.2cm}d\mu)$ with $M$ is a smooth
manifold and the operator $T$ is a self-adjoint  elliptic partial differential operator, the spectral density
$e_{\lambda}=\frac{dE_{\lambda}}{d\lambda}$ has an associated Schwartz kernel $e(\lambda,z,w)$ being an element
of $D^{'}(\mathbb{R},D^{'}(M\times M))$ \cite{Est}.\\
Note that it is not hard to see that the operator $\tilde{\Delta}_{\nu}$ defined by $(\ref{E2.12})$ and $(\ref{E2.13})$
 is a self-adjoint elliptic operator \cite{Ask}. Then, based on the formula $(\ref{E2.22})$ for the spectral projectors,
  we can prove the following result.
\begin{proposition}\label{P2.2}
The Schwartz kernel of the spectral density associated with the self-adjoint elliptic operator  $\tilde{\Delta}_{\nu}$ is given by
\begin{align}\label{E2.26}
e^{\nu}(\lambda,z,w)=\frac{\pi}{\nu}\sum_{j=0}^{+\infty}e^{\nu z\overline{w}}L_{j}(\nu\mid z-w\mid^{2})\delta(\lambda-\nu j),
\end{align}
\end{proposition}
where $\delta$ is the Dirac delta distribution.
\begin{proof}
By using the point $(3)$ of the proposition $(\ref{P2.1})$, it is not hard to see that the orthogonal projector on the
eigenspace $\mathcal{E}_{\ell}^{2,\nu}(\mathbb{C})$ defined in $(\ref{E2.18})$ is the integral operator defined by
\begin{align}\label{E2.27}
\nonumber P_{\ell}^{\nu}:\hspace{0.25cm}L^{2}(\mathbb{C},&\hspace{0.2cm}e^{-\nu\mid z\mid^{2}}d\lambda(z))\longrightarrow \mathcal{E}_{\ell}^{2,\nu}(\mathbb{C})\\
        &\varphi\longmapsto P_{\ell}^{\nu}[\varphi](z)=\int_{\mathbb{C}}\varphi(w)\mathcal{K}_{\ell}^{\nu}(z,w)e^{-\nu\mid w\mid^{2}}d\lambda(w),
\end{align}
where $\mathcal{K}_{\ell}^{\nu}(z,w)$ is the reproducing kernel of $\mathcal{E}_{\ell}^{2,\nu}(\mathbb{C})$.\\
Applying the formula \cite[p.240]{Mag}
\begin{align}\label{E2.28}
\sum_{j=0}^{k}L_{j}^{(\alpha)}(x)=L_{k}^{(\alpha+1)}(x),
\end{align}
for $\alpha=0$, $k=[\lambda]$, $\lambda\geq0$ and $x=\nu\mid z-w\mid^{2}$, we can see that the spectral projector
 $E_{\lambda}^{\nu}$ defined in $(\ref{E2.21})$ can be rewritten as
\begin{align}\label{E2.29}
E_{\lambda}^{\nu}=
\left\{
  \begin{array}{ll}
    \displaystyle{\sum_{\ell=0}^{[\lambda]}}P_{\ell}^{\nu}, &\lambda\geq0, \hbox{} \\\\
    0, &\lambda<0. \hbox{}
  \end{array}
\right.
\end{align}
Notice that the values $\lambda_{\ell}=\nu\ell$ ($\ell\in \mathbb{Z}_{+}$) are the only discontinuity points
 of the vector valued function $\lambda\longmapsto E_{\lambda}^{\nu}$ for which the corresponding jumps are given by
 \begin{align}\label{E2.30}
E_{\lambda_{\ell}+0}^{\nu}-E_{\lambda_{\ell}-0}^{\nu}=P_{\ell}^{\nu}.
\end{align}
In addition, $\lambda\longmapsto E_{\lambda}^{\nu}$ is constant on the interval $]\lambda_{\ell},\lambda_{\ell+1}[$,
 for each $\ell\in \mathbb{Z}_{+}$.\\
Now, we are in position to give the following derivation formula of $\lambda\longmapsto E_{\lambda}^{\nu}$ in the distribution sense.
 Namely, we have
\begin{align}\label{E2.31}
\frac{d E_{\lambda}^{\nu}}{d\lambda}=\sum_{\ell=0}^{+\infty}(E_{\lambda_{\ell}+0}^{\nu}-E_{\lambda_{\ell}-0}^{\nu})
\delta(\lambda-\nu\ell).
\end{align}
Combining $(\ref{E2.27})$ and $(\ref{E2.30})$, we show that the Schwartz kernel of the distribution $\frac{d E_{\lambda}^{\nu}}{d\lambda}$ is given by
\begin{align}\label{E2.32}
e^{\nu}(\lambda,z,w)=\frac{\pi}{\nu}\sum_{j=0}^{+\infty}e^{\nu z\overline{w}}L_{j}(\nu\mid z-w\mid^{2})\delta(\lambda-\nu j).
\end{align}
\end{proof}
Note that the Schwartz kernel of the spectral density associated with a self adjoint elliptic operator $T$ acting on a
complex Hilbert space $L^{2}(M,d\mu)$ can be used to build a functional calculus for the operator $T$ \cite{Est}. Precisely, for
 a suitable Borelian function $f\in\mathcal{B}(\mathbb{R})$, we have the following formula
\begin{align}\label{E2.33}
f(T)[\varphi](z)=\int_{M}\Omega_{f}(w,z)\varphi(w)d\mu(w),
\end{align}
where the kernel $\Omega_{f}(w,z)$ is defined by
\begin{align}\label{E2.34}
\Omega_{f}(w,z)=\int_{\sigma(T)}e(\lambda,z,w)f(\lambda)d\lambda,
\end{align}
and $\sigma(T)$ is the spectrum of the operator $T$. The equations $(\ref{E2.33})$ and $(\ref{E2.34})$ are
 understood in the distribution sense with Gel'fand-Shilov notation \cite{Gel}.\\
As application of the proposition $(\ref{P2.2})$, we have the following corollary.
\begin{corollary}\label{C2.1}
\begin{enumerate}
  \item The distribution integral kernel associated with the heat propagator $e^{-t\tilde{\Delta}_{\nu}}$ (heat semigroup) is given by
\begin{align}\label{E2.35}
K_{\nu}(t,z,w)=\frac{\pi}{\nu}e^{\nu z\overline{w}}\sum_{k=0}^{+\infty}e^{-\nu tk}L_{k}(\nu\mid z-w\mid^{2}),
\hspace{0.2cm}t>0\hspace{0.2cm}\mbox{and}\hspace{0.2cm}z,\hspace{0.1cm}w\in \mathbb{C}.
\end{align}
\item The distribution integral kernel associated with the Schr\"{o}dinger propagator $e^{-it\tilde{\Delta}_{\nu}}$ (dynamical group)
 is given by
\begin{align}\label{E2.36}
G_{\nu}(t,z,w)=\frac{\pi}{\nu}e^{\nu z\overline{w}}\sum_{k=0}^{+\infty}e^{-i\nu tk}L_{k}(\nu\mid z-w\mid^{2}),
\hspace{0.2cm}t\in \mathbb{R}\hspace{0.2cm}\mbox{and}\hspace{0.2cm}z, \hspace{0.1cm}w\in \mathbb{C}.
\end{align}
\end{enumerate}
\end{corollary}
\begin{proof}
According to the formulas $(\ref{E2.33})$ and $(\ref{E2.34})$, the distribution integral kernel associated with
the Borelian function $f_{t}(\lambda)=e^{-\lambda t}$ is given by
\begin{align}\label{E2.37}
\nonumber K_{\nu}(t,z,w)&=\Omega_{e^{-\lambda t}}(w,z)\\
&=\frac{\pi}{\nu}e^{\nu z\overline{w}}\sum_{k=0}^{+\infty}e^{-\nu tk}L_{k}(\nu\mid z-w\mid^{2}),\hspace{0.2cm}t>0.
\end{align}
The above series has also for $t>0$ a pointwise sense. Indeed, it is a convergent series and its sum
 is given by the classical generating  function for the Laguerre polynomials \cite[p.114]{Bea}
\begin{align}\label{E2.38}
\sum_{k=0}^{+\infty}L_{k}^{(\beta)}(x)s^{k}=(1-s)^{-\beta-1}\exp(\frac{-xs}{1-s}),
\end{align}
for $\beta=0$, $x=\nu\mid z-w\mid^{2}$ and $s=e^{-\nu t}$. Hence,
the kernel $K_{\nu}(t,z,w)$ can be also written as
\begin{align}\label{E2.39}
K_{\nu}(t,z,w)=\frac{\pi e^{\nu z\overline{w}}}{\nu(1-e^{-\nu t})}\exp(\frac{-\nu \mid z-w \mid^{2}e^{-\nu t}}{1-e^{-\nu t}}).
\end{align}
Then, $(1)$ is proved.\\
Also, by using the formulas $(\ref{E2.33})$ and $(\ref{E2.34})$, we see that the distribution integral kernel associated with the Borelian function $f_{t}(\lambda)=e^{-i\lambda t}$ is given by
\begin{align}\label{E2.40}
\nonumber G_{\nu}(t,z,w)&=\Omega_{e^{-i\lambda t}}(w,z)\\
&=\frac{\pi}{\nu}e^{\nu z\overline{w}}\sum_{k=0}^{+\infty}e^{-i\nu tk}L_{k}(\nu\mid z-w\mid^{2}).
\end{align}
The above series is well defined as a distribution in $\mathcal{D}^{'}(\mathbb{R}_{t})$. Indeed, by using the asymptotic formula for
Laguerre polynomials \cite[p.248]{Mag}
\begin{align}\label{E2.41}
 L_{k}^{(\beta)}(x)=x^{\frac{-\beta}{2}-\frac{1}{4}}O(k^{\frac{\beta}{2}-\frac{1}{4}}),\hspace{0.15cm}\mbox{as}\hspace{0.15cm} k\longrightarrow+\infty,\hspace{0.15cm} \mbox{for} \hspace{0.15cm}\frac{c}{k}\leq x\leq w,
 \end{align}
where $c$ and $w$ are fixed positive constants,  we obtain, for
$\beta=0$ and $x=\nu\mid z-w\mid^{2}>0$,  the following estimate
\begin{align}\label{E2.42}
\mid L_{k}(\nu\mid z-w\mid^{2})\mid\leq M_{\nu}\mid z-w\mid^{-\frac{1}{2}}k^{-\frac{1}{4}},
 \end{align}
where $M_{\nu}$ is a positive constant and $k$ is enough large. This above inequality leads us to obtain
for the general term of the involved series in $(\ref{E2.40})$, the following behavior
\begin{align}\label{E2.43}
\mid e^{-\nu t k}L_{k}(\nu\mid z-w\mid^{2})\mid\leq A \mid k\mid^{m}+B,\hspace{0.2cm}z\neq w\hspace{0.2cm}\mbox{as}\hspace{0.2cm}k\longrightarrow+\infty,
 \end{align}
for $A=M_{\nu}\mid z-w\mid^{-\frac{1}{2}}$, $B=0$ and $m=0$. In the case $z=w$, the inequality $(\ref{E2.43})$ stay true with $A=0$, $B=1$ and $m$ is any fixed integer (we have used the fact $L_{k}(0)=1$ (\cite[p.240]{Mag})). Then, according to \cite[p.97]{Vla}, the kernel $G_{\nu}(t,z,w)$ is well defined as an element of $\mathcal{D}^{'}(\mathbb{R}_{t})$.  Hence, the point $(2)$ is proved. This closes the proof of the corollary.
\end{proof}
The ellipticity of the operator $\tilde{\Delta}_{\nu}$ assures that any eigenfunction of the operator $\tilde{\Delta}_{\nu}$
 is of class $C^{\infty}$ (see \cite{Hor} for the general theory). Furthermore, the eigenstate
  $\psi\in \mathcal{E}_{0}^{2,\nu}(\mathbb{C})$ corresponding to the lower energy $0$ involves  a
  complex regularity. Precisely, we have the following proposition.
\begin{proposition}\label{P2.3}
The harmonic space
\begin{align}\label{E2.44}
\mathcal{E}_{0}^{2,\nu}(\mathbb{C})=\{F\in D(\tilde{\Delta}_{\nu}),\hspace{0.2cm}\tilde{\Delta}_{\nu}F=0\},
\end{align}
of the self-adjoint operator $\tilde{\Delta}_{\nu}$ defined by $(\ref{E2.12})$ and $(\ref{E2.13})$ coincides with the Bargmann-Fock space on $\mathbb{C}$
\begin{align}\label{E2.45}
\mathcal{A}^{2,\nu}(\mathbb{C})=\{f:\mathbb{C}\longrightarrow\mathbb{C}, \hspace{0.15cm} \mbox{holomorphic and}
 \hspace{0.15cm}\int_{\mathbb{C}}\mid f(z)\mid^{2}e^{-\nu\mid z\mid^{2}}d\lambda(z)<+\infty\}.
\end{align}
\end{proposition}
\begin{proof}
First, recall that the Bargmann-Fock space $\mathcal{A}^{2,\nu}(\mathbb{C})$ is a R.K.H.S with the reproducing kernel \cite{Bar}
\begin{align}\label{E2.46}
\mathcal{K}_{\nu}(z,w)=\frac{\nu}{\pi}e^{\nu z\overline{w}}.
\end{align}
From the point $(3)$ of the proposition $(\ref{P2.1})$, the reproducing kernel of the eigenspace $\mathcal{E}_{0}^{2,\nu}(\mathbb{C})$ is given by
\begin{align}\label{E2.47}
\mathcal{K}^{\nu}_{0}(z,w)=\frac{\nu}{\pi}e^{\nu z\overline{w}}L_{0}(\nu\mid z-w\mid^{2}).
\end{align}
Using the fact that $L_{0}(x)\equiv1$, we see that the eigenstate $\mathcal{E}_{0}^{2,\nu}(\mathbb{C})$ admits
 also the same reproducing kernel as the Bargmann-Fock space $\mathcal{A}^{2,\nu}(\mathbb{C})$. Thus, using the proposition $(3.3)$
  in \cite{Pau}, we get the desired result.
\end{proof}
Here, the partial differential operator
$-\frac{\partial^{2}}{\partial z\partial\overline{z}}+\nu\overline{z}\frac{\partial}{\partial\overline{z}}$
will plays an important role in the characterization of the Bargmann-Dirichlet space $\mathcal{A}^{2,\nu}_{1}(\mathbb{C})$
as harmonic space of a second order elliptic partial differential operator. Precisely, we have the following proposition
\begin{proposition}\label{P2.4}
Let $\nu>0$ and let $\Delta_{\nu}$ be the partial differential operator defined by
\begin{align}\label{E2.48}
\Delta_{\nu}=-\frac{\partial^{2}}{\partial z\partial\overline{z}}+\nu\overline{z}\frac{\partial}{\partial\overline{z}},
\end{align}
acting on the Hilbert space  $L^{2,\nu}(\mathbb{C}):=L^{2}(\mathbb{C},\hspace{0.2cm}e^{-\nu\mid z\mid^{2}}d\lambda(z))$ with the dense domain
\begin{align}\label{E2.49}
D(\Delta_{\nu})=\{F\in L^{2,\nu}(\mathbb{C}),\hspace{0.25cm}\Delta_{\nu}F\in L^{2,\nu}(\mathbb{C})\hspace{0.25cm}\mbox{and}\hspace{0.25cm}\frac{\partial F}{\partial z}\in L^{2,\nu}(\mathbb{C})\}.
\end{align}
Then, the Bargmann-Dirichlet space $\mathcal{A}^{2,\nu}_{1}(\mathbb{C})$ is the null space of $\Delta_{\nu}$, that is,
\begin{align}\label{E2.50}
\mathcal{A}^{2,\nu}_{1}(\mathbb{C})=\{F\in D(\Delta_{\nu}),\hspace{0.25cm}\Delta_{\nu}F=0\}.
\end{align}
\end{proposition}
\begin{proof} Let $F\in D(\Delta_{\nu})$ such that $\Delta_{\nu}F=0$. Then, we have
\begin{center}
$\textbf{(i)}$  $F\in L^{2,\nu}(\mathbb{C}),$ $\textbf{(ii)}$ $\Delta_{\nu}F=0$ and
$\textbf{(iii)}$  $\frac{\partial F}{\partial z}\in L^{2,\nu}(\mathbb{C})$.
\end{center}
By the proposition $(\ref{P2.3})$, we see that both conditions $\textbf{(i)}$ and $\textbf{(ii)}$ imply that $F$ is holomorphic.
Then, by combining this property with the fact $\frac{\partial F}{\partial z}\in L^{2,\nu}(\mathbb{C})$, we get that $F$ belongs to
the Bargmann-Dirichlet space $\mathcal{A}^{2,\nu}_{1}(\mathbb{C})$. Inversely, let $F(z):=\displaystyle{\sum_{j=0}^{+\infty}} a_{j}z^{j}$ be
 an element of the space $\mathcal{A}_{1}^{2,\nu}(\mathbb{C})$ defined in $(\ref{E2.3})$. By using the polar coordinates,  we obtain
\begin{align}\label{E2.51}
\int_{\mathbb{C}}\mid F(z)\mid^{2}e^{-\nu\mid z\mid^{2}}d\lambda(z)=
2\pi\int_{0}^{+\infty}[\int_{0}^{2\pi}\mid\sum_{j=0}^{+\infty} a_{j}r^{j}e^{ij\theta}\mid^{2}\frac{d\theta}{2\pi}]
e^{-\nu r^{2}}rdr.
\end{align}
By Parseval's formula, for $r\in[0,1]$, we get
\begin{align}\label{E2.52}
\int_{0}^{2\pi}\mid\sum_{j=0}^{+\infty} a_{j}r^{j}e^{ij\theta}\mid^{2}\frac{d\theta}{2\pi}=\sum_{j=0}^{+\infty}r^{2j}\mid a_{j}\mid^{2}.
\end{align}
Inserting the last equality in the right hand side of the equation $(\ref{E2.51})$ and using the monotone convergence theorem,
the equation $(\ref{E2.51})$ becomes
\begin{align}\label{E2.53}
\int_{\mathbb{C}}\mid F(z)\mid^{2}e^{-\nu\mid z\mid^{2}}d\lambda(z)=
2\pi\sum_{j=0}^{+\infty} \mid a_{j}\mid^{2}\int_{0}^{+\infty} r^{2j+1}e^{-\nu r^{2}}dr.
\end{align}
Using the change of variable $s=\nu r^{2}$, the involved integral in the right hand side of the equation $(\ref{E2.53})$
takes the following form
\begin{align}\label{E2.54}
\nonumber \int_{0}^{+\infty} r^{2j+1}e^{-\nu r^{2}}dr&=\frac{1}{2\nu^{j+1}}\int_{0}^{+\infty}s^{j}e^{-s}ds\\
\nonumber&=\frac{1}{2\nu^{j+1}} \Gamma(j+1)\\
&=\frac{1}{2\nu^{j+1}}j!,
\end{align}
where $\Gamma$ is the classical Euler Gamma function. Thus, the equation $(\ref{E2.53})$ becomes
\begin{align}\label{E2.55}
\int_{\mathbb{C}}\mid F(z)\mid^{2}e^{-\nu\mid z\mid^{2}}d\lambda(z)=\frac{\pi}{\nu}\sum_{j=0}^{+\infty} \frac{j!}{\nu^{j}}\mid a_{j}\mid^{2}.
\end{align}
Consider the following inequality
\begin{align}\label{E2.56}
\sum_{j=1}^{+\infty} \frac{1}{\nu^{j+1}}j!\mid a_{j}\mid^{2}\leq\frac{1}{\nu}\sum_{j=1}^{+\infty} \frac{1}{\nu^{j}}j(j!)\mid a_{j}\mid^{2}.
\end{align}
The membership test given in $(\ref{E2.11})$  ensures that the last series is convergent. This with the inequality  $(\ref{E2.56})$ imply
that $F$ is in $ L^{2,\nu}(\mathbb{C})$. The holomorphicity of $F$ proves that $\frac{\partial F}{\partial \overline{z}}=0$. Then, we obtain
\begin{align}\label{E2.57}
-\frac{\partial^{2} F}{\partial z\partial\overline{z}}+\nu\overline{z}\frac{\partial F}{\partial\overline{z}}=0, \hspace{0.25cm}\nu>0,
\end{align}
Thus, one has $F\in L^{2,\nu}(\mathbb{C}),$ $\frac{\partial F}{\partial z}\in L^{2,\nu}(\mathbb{C})$ and  $\Delta_{\nu}F=0.$\\ Hence,
 we get $F\in \{F\in D(\Delta_{\nu}),\hspace{0.25cm}\Delta_{\nu}F=0\}.$ The proof of the proposition is closed.
\end{proof}
Notice that the Bargmann-Fock space $\mathcal{A}^{2,\nu}(\mathbb{C})$ defined in $(\ref{E2.45})$ is not stable by any derivative
 order $\frac{\partial^{m}(.)}{\partial z^{m}}$. Precisely, we have the following lemma.
\begin{lemma}\label{L2.1} Let $\nu>0$ and $m\in \mathbb{Z}_{+}$, $m\geq1$. Then, for the holomorphic function defined by
\begin{align}\label{E2.58}
\varphi_{\nu}(z)=\sum_{j=0}^{+\infty}\frac{\nu^{\frac{j+1}{2}}}{\sqrt{(j+2)!}}z^{j},
\end{align}
 we have the following properties:
\begin{enumerate}
  \item $\varphi_{\nu}(z)$ belongs to the Bargmann-Fock space $\mathcal{A}^{2,\nu}(\mathbb{C})$ defined in $(\ref{E2.45})$.
  \item $\frac{\partial^{m} \varphi_{\nu}(z)}{\partial z^{m}}$ does not belong to the Bargmann-Fock space $\mathcal{A}^{2,\nu}(\mathbb{C})$, for all $m\geq1$.
\end{enumerate}
\end{lemma}

\begin{proof}
By using the asymptotic Stirling formula \cite[p.312]{Nik}
\begin{align}\label{E2.59}
j!\sim \sqrt{2\pi j}\bigg(\frac{j}{e}\bigg)^{j},\hspace{0.2cm}\mbox{as}\hspace{0.2cm}j\longrightarrow+\infty,
\end{align}
we can see easily that the series in $(\ref{E2.58})$ is an entire series with infinite radius. Thus, the holomorphicity of $\varphi_{\nu}$ follows.
 Applying the formula in  $(\ref{E2.55})$ to the above function $\varphi_{\nu}$ defined in $(\ref{E2.58})$, we obtain

\begin{align}\label{E2.60}
\int_{\mathbb{C}}\mid \varphi_{\nu}(z)\mid^{2}e^{-\nu\mid z\mid^{2}}d\lambda(z)=\frac{\pi}{\nu}\sum_{j=0}^{+\infty}\frac{j!}{\nu^{j}}
\mid\frac{\nu^{\frac{j+1}{2}}}{\sqrt{(j+2)!}}\mid^{2}=
\pi\sum_{j=0}^{+\infty}\frac{1}{(j+1)(j+2)}<+\infty.
\end{align}
Using the convergence of the last series, we assure that the function $\varphi_{\nu}$ belongs to the Bargmann-Fock
space $\mathcal{A}^{2,\nu}(\mathbb{C})$.\\
Now, using the second theorem of Weirstrass \cite[p.115]{Cha}, on each compact disk $\overline{D}(0,R)=\{z\in\mathbb{C},\hspace{0.2cm}
\mid z\mid\leq R\}$, $R>0$, the complex derivative of order $m\geq1$ for the entire series $\varphi_{\nu}(z)$ can be obtained by deriving
term-by-term. Then, for each $z\in\mathbb{C}=\displaystyle{\cup_{R>0}}\overline{D}(0,R)$, we have
\begin{align}\label{E2.61}
\nonumber \frac{\partial^{m} \varphi_{\nu}(z)}{\partial z^{m}}&=\sum_{j=m}^{+\infty}\frac{\nu^{\frac{j+1}{2}}j(j-1)...(j-(m-1))}{\sqrt{(j+2)!}}z^{j-m}\\
&=\sum_{j=0}^{+\infty}\frac{\nu^{\frac{j+m+1}{2}}(j+m)(j+m-1)...(j+1)}{\sqrt{(j+m+2)!}}z^{j}.
\end{align}
Also, by applying again the formula $(\ref{E2.55})$ to the above holomorphic series, we obtain
\begin{align}\label{E2.62}
\nonumber \int_{\mathbb{C}}\mid \frac{\partial^{m}\varphi_{\nu}(z)}{\partial z^{m}}\mid^{2}e^{-\nu\mid z\mid^{2}}
d\lambda(z)&=\frac{\pi}{\nu}\sum_{j=0}^{+\infty}\frac{j!}{\nu^{j}}\mid\frac{\nu^{\frac{j+m+1}{2}}(j+m)(j+m-1)...(j+1)}{\sqrt{(j+m+2)!}}\mid^{2}\\
&=\pi\nu^{m}\sum_{j=0}^{+\infty}\frac{(j+m)^{2}(j+m-1)^{2}...(j+1)^{2}}{(j+1)(j+2)...(j+m+2)}.
\end{align}
Using the fact that the general member of the last series is equivalent to the term $\frac{\pi\nu^{m}}{(j+1)^{2-m}}$.
Then, the series in $(\ref{E2.62})$ has the same convergence nature as the Riemann series $\displaystyle{\sum_{j=0}^{+\infty}}\frac{1}{(j+1)^{2-m}}$,
which is convergent if and only if $2-m>1$. This implies that $m<1$. Hence, the holomorphic function $\frac{\partial^{m} \varphi_{\nu}(z)}{\partial z^{m}}$
 does not belong to the Bargmann-Fock space $\mathcal{A}^{2,\nu}(\mathbb{C})$ for all $m\geq1$.
\end{proof}
Now, we give some spectral properties for the operator $\Delta_{\nu}$ defined by $(\ref{E2.48})$ and  $(\ref{E2.49})$. Precisely, by using
 the same notations as in the proposition $(\ref{P2.4})$, we can state the following proposition.
\begin{proposition}\label{P2.5} We have the following properties:
\begin{enumerate}
 \item The operator $\Delta_{\nu}$ is closable and admits a self-adjoint extension.
  \item The operator $\Delta_{\nu}$ is an unbounded non self-adjoint operator.
  \item $0$ belongs to the point spectrum of $\Delta_{\nu}$.
\end{enumerate}
\end{proposition}
\begin{proof}
We return back to the self-adjoint operator $\tilde{\Delta}_{\nu}$ defined by $(\ref{E2.12})$ and  $(\ref{E2.13})$. It is easy to see that the operator
 $\tilde{\Delta}_{\nu}$ is an extension of the operator  $\Delta_{\nu}$ (see \cite[p.4]{Kon} for the general theory). We write this fact by
\begin{align}\label{E2.63}
\Delta_{\nu}\subset \tilde{\Delta}_{\nu}.
\end{align}
Thus, $(1)$ is proved.\\
For proving $(2)$, first we shall prove that $D(\Delta_{\nu})\neq
D(\tilde{\Delta}_{\nu})$. To do so,
 we consider the holomorphic function $\varphi_{\nu}(z)$ defined by $(\ref{E2.58})$ in lemma $(\ref{L2.1})$.
 By using the point $1$ of the last lemma, we have $\varphi_{\nu}$ is holomorphic and
  $\varphi_{\nu}\in L^{2}(\mathbb{C},\hspace{0.2cm}e^{-\nu\mid z\mid^{2}}d\lambda(z))$. Then, $\frac{\partial \varphi_{\nu}}{\partial \overline{z}}=0$ and
   so we obtain
$-\frac{\partial^{2} \varphi_{\nu}}{\partial z\partial\overline{z}}+\nu\overline{z}\frac{\partial \varphi_{\nu}}{\partial\overline{z}}=0$.
Hence, $\varphi_{\nu}$ belongs to the domain $D(\tilde{\Delta}_{\nu})$ defined in $(\ref{E2.13})$. From the point $2$ of the Lemma $(\ref{L2.1})$,
 we have $\varphi_{\nu}:=\frac{\partial \varphi_{\nu}}{\partial z}\notin L^{2}(\mathbb{C},\hspace{0.2cm}e^{-\nu\mid z\mid^{2}}d\lambda(z))$.
  This proves that $\varphi_{\nu}\notin D(\Delta_{\nu})$ defined in $(\ref{E2.49})$. Hence, it is remarked that  $D(\Delta_{\nu})\neq D(\tilde{\Delta}_{\nu})$.
Now, if we suppose that $\Delta_{\nu}$ is self-adjoint, we obtain from $(\ref{E2.63})$ and the proposition $(1.6)$ in \cite[p.9]{Kon} that
\begin{align}\label{E2.64}
 \tilde{\Delta}_{\nu}=(\tilde{\Delta}_{\nu})^{*} \subset \Delta_{\nu}^{*}=\Delta_{\nu}.
\end{align}
Using $(\ref{E2.63})$ and $(\ref{E2.64})$, we get $\Delta_{\nu}=\tilde{\Delta}_{\nu}$, which implies that
$D(\Delta_{\nu})=D(\tilde{\Delta}_{\nu})$. This contradicts the fact $D(\Delta_{\nu})\neq D(\tilde{\Delta}_{\nu})$.
 Thus, $(2)$ is proved. The point $(3)$ is just an other way to state the proposition $(\ref{P2.4})$. The proof is closed.
\end{proof}
\begin{remark}
 The operator $\Delta_{\nu}$ is not closed.
\end{remark}
Indeed, let us consider the following sequence of the holomorphic polynomials
\begin{align}\label{E2.65}
\varphi_{\nu,k}(z)=\sum_{j=0}^{k}\frac{\nu^{\frac{j+1}{2}}}{\sqrt{(j+2)!}}z^{j}.
\end{align}
It is easy to see that $\varphi_{\nu,k}$ belongs to the domain
$D(\Delta_{\nu})$ defined in $(\ref{E2.49})$. By using the formula
$(\ref{E2.55})$,
 it is not hard to show that
\begin{align}\label{E2.66}
\nonumber \parallel \varphi_{\nu,k}-\varphi_{\nu}\parallel_{\nu}^{2}&=\int_{\mathbb{C}}\mid \varphi_{\nu,k}(z)-\varphi_{\nu}(z)\mid^{2}e^{-\nu\mid z\mid^{2}}d\lambda(z)\\
&=\pi\sum_{j=k+1}^{+\infty}\frac{1}{(j+1)(j+2)},
\end{align}
where $\varphi_{\nu}$ is the function defined in $(\ref{E2.58})$.
Thank's to the convergence of the series
$\displaystyle{\sum_{j=0}^{+\infty}}\frac{1}{(j+1)(j+2)}$, we see
easily that the series appearing in $(\ref{E2.66})$ goes
  to zero as $k\longrightarrow+\infty$. From the lemma $(\ref{L2.1})$,  it has been shown that  $\frac{\partial \varphi_{\nu}}{\partial z}$
   does not belong to the Hilbert space $L^{2}(\mathbb{C},\hspace{0.2cm}e^{-\nu\mid z\mid^{2}}d\lambda(z))$. This last fact implies
    that $\varphi_{\nu}$ is not in the domain $D(\Delta_{\nu})$. Thus, we have
\begin{align}\label{E2.67}
\left\{
  \begin{array}{ll}
    \varphi_{\nu,k}\in D(\Delta_{\nu}),  \;\; \displaystyle{\lim_{k\longrightarrow+\infty}}\varphi_{\nu,k}=\varphi_{\nu}\hspace{0.2cm}
     \mbox{in} \hspace{0.2cm} L^{2}(\mathbb{C},\hspace{0.2cm}e^{-\nu\mid z\mid^{2}}d\lambda(z)),  & \hbox{} \\
\\
    \Delta_{\nu}\varphi_{\nu,k}=0\longrightarrow0,\hspace{0.2cm}
     \mbox{as} \hspace{0.2cm}k\longrightarrow+\infty & \hbox{} \\
\\
    \varphi_{\nu}\notin D(\Delta_{\nu}). & \hbox{}
  \end{array}
\right.
\end{align}
Hence, the above tree properties assures that the operator $\Delta_{\nu}$ is not closed.\\

\subsection{Bargmann transform associated with the Bargmann-Dirichlet space}
Now, we give a Bargmann transform associated with the Bargmann-Dirichlet space.
\begin{theorem}\label{T2.1}
For the Bargmann-Dirichlet space $\mathcal{A}_{1}^{2,\nu}(\mathbb{C})$, we have the following associated Bargmann-Dirichlet transform
\begin{align}\label{E2.68}
 \nonumber B_{\nu}\hspace{0.2cm}:L^{2}(\mathbb{R},&\hspace{0.2cm} dx)\longrightarrow \mathcal{A}_{1}^{2,\nu}(\mathbb{C})\\
 &\varphi\longmapsto B_{\nu}[\varphi](z):=\int_{-\infty}^{+\infty}K_{\nu}(z,x)\varphi(x)dx,
 \end{align}
  where the integral kernel $K_{\nu}(z,x)$ is given by
 \begin{align}\label{E2.69}
 K_{\nu}(z,x)=\sqrt{\nu}(\pi)^{\frac{-3}{4}}e^{\frac{-x^{2}}{2}}[1+\frac{\sqrt{2}}{\sqrt{\pi}}z\int_{0}^{+\infty}\sqrt{t}e^{-t}
\exp(x\sqrt{2\nu}e^{-t}z-\nu e^{-2t}\frac{z^{2}}{2})H_{1}(x-\sqrt{\frac{\nu}{2}}e^{-t}z)dt],
\end{align}
with $H_{1}(x)=2x$ denote the second classical Hermite polynomial.
\end{theorem}
\begin{proof} We will give our proof in two steps. In the first step, we consider the kernel function defined on $\mathbb{C}\times\mathbb{R}$ by
\begin{align}\label{E2.70}
 \tilde{K}_{\nu}(z,x)=\sum_{j=0}^{+\infty}\varphi_{j}(x)\psi_{j}(z),
\end{align}
associated with the orthonormal basis $\{\psi_{j}\}_{i\in \mathbb{Z}_{+}}$, that we have mentioned in $(\ref{E2.7})$ and
the orthonormal basis of $L^{2}(\mathbb{R},\hspace{0.2cm}e^{-x^{2}}dx)$ defined by means of the orthonormalized
 Hermite polynomials \cite[p.109]{Bea}
\begin{align}\label{E2.71}
 \varphi_{j}(x)=\frac{1}{\pi^{\frac{1}{4}}\sqrt{2^{j}j!}}H_{j}(x),\hspace{0.25cm}j\in \mathbb{Z}_{+}.
\end{align}
First, we have to compute the above kernel $\tilde{K}_{\nu}(z,x)$. To do so, we rewrite the kernel $\tilde{K}_{\nu}(z,x)$
in terms of the expression of $\psi_{j}(z)$ and $\varphi_{j}(x)$ as follows
\begin{align}\label{E2.72}
 \nonumber \tilde{K}_{\nu}(z,x)&=(\pi)^{\frac{-3}{4}}[\sqrt{\nu}+
 \sum_{j=1}^{+\infty}\frac{H_{j}(x)}{\sqrt{2^{j}j!}}\frac{\nu^{\frac{j}{2}}}{\sqrt{j(j!)}}z^{j}]\\
\nonumber&=(\pi)^{\frac{-3}{4}}[\sqrt{\nu}+\sum_{j=1}^{+\infty}
\frac{H_{j}(x)}{j!}\frac{\nu^{\frac{j}{2}}}{2^{\frac{j}{2}}}\frac{z^{j}}{\sqrt{j}}]\\
\nonumber&=(\pi)^{\frac{-3}{4}}[\sqrt{\nu}+\sqrt{\frac{\nu}{2}}z\sum_{k=0}^{+\infty}
\frac{H_{k+1}(x)}{(k+1)!\sqrt{k+1}}(\sqrt{\frac{\nu}{2}}z)^{k}]\\
&=\sqrt{\nu}(\pi)^{\frac{-3}{4}}[1+\frac{z}{\sqrt{2}}T_{\nu}(z,x)],
\end{align}
where
\begin{align}\label{E2.73}
 T_{\nu}(z,x)=\sum_{k=0}^{+\infty}
\frac{H_{k+1}(x)}{(k)!(k+1)^{\frac{3}{2}}}(\sqrt{\frac{\nu}{2}}z)^{k}.
\end{align}
Note that (up to our knowledge) the above series does not appear in the literature as a standard closed generating
 formula for the Hermite polynomials. To avoid this problem, we consider the following integral representation \cite[p.42]{Sch}
\begin{align}\label{E2.74}
 \frac{1}{\lambda^{s}}=\frac{1}{\Gamma(s)}\int_{0}^{+\infty}e^{-\lambda t}t^{s-1}dt, \hspace{0.25cm}Re(\lambda)>0,
\end{align}
for $s=\frac{3}{2}$ and $\lambda=k+1$. Then, we can write
\begin{align}\label{E2.75}
 \frac{1}{(k+1)^{\frac{3}{2}}}=\frac{1}{\Gamma(\frac{3}{2})}\int_{0}^{+\infty}e^{-k t}(e^{-t}\sqrt{t})dt.
\end{align}
By inserting the above formula in the equation $(\ref{E2.73})$, we obtain
\begin{align}\label{E2.76}
 T_{\nu}(z,x)=\frac{1}{\Gamma(\frac{3}{2})}\sum_{k=0}^{+\infty}\frac{H_{k+1}(x)}{k!}
\int_{0}^{+\infty}e^{-k t}(e^{-t}\sqrt{t})(\sqrt{\frac{\nu}{2}}z)^{k}dt.
\end{align}
To compute the  function $T_{\nu}(z,x)$, we need to permute the integral and the sum.
To do so, we first recall the following asymptotic formula for the Hermite polynomials \cite[p.112 and p.338]{Bea}
\begin{align}\label{E2.77}
 H_{k}(x)=2^{\frac{k}{2}}\frac{2^{\frac{1}{4}}(k!)^{\frac{1}{2}}}{(k\pi)^{\frac{1}{4}}}e^{\frac{x^{2}}{2}}
[\cos(\sqrt{2k+1}x-\frac{k\pi}{2})+O(k^{\frac{-1}{2}})],\hspace{0.2cm}as\hspace{0.1cm}k\longrightarrow+\infty.
\end{align}
This formula help us to obtain, for each fixed $x$, that
\begin{align}\label{E2.78}
 \mid H_{k}(x)\mid\leq C(x)\frac{2^{\frac{k}{2}}(k!)^{\frac{1}{2}}}{k^{\frac{1}{4}}},
 \hspace{0.2cm}\mbox{for}\hspace{0.2cm} k\hspace{0.2cm} \mbox{enough large}.
\end{align}
Let $p_{0}$ be a fixed integer enough large and let $p\geq p_{0}$, then we have the following inequality
\begin{align}\label{E2.79}
 \mid\sum_{k=0}^{p}\frac{H_{k+1}(x)}{k!}
e^{-k t}(\sqrt{\frac{\nu}{2}}z)^{k}\mid\leq
\sum_{k=0}^{p_{0}}\frac{\mid H_{k+1}(x)\mid}{k!}
\mid\sqrt{\frac{\nu}{2}}z\mid^{k}
+\sum_{k=p_{0}}^{p}\frac{\mid H_{k+1}(x)\mid}{k!}
\mid\sqrt{\frac{\nu}{2}}z\mid^{k},\hspace{0.2cm}t\geq0.
\end{align}
By using $(\ref{E2.78})$, we obtain for the last sum in the right hand side of $(\ref{E2.79})$
\begin{align}\label{E2.80}
\sum_{k=p_{0}}^{p}\frac{\mid H_{k+1}(x)\mid}{k!}
\mid\sqrt{\frac{\nu}{2}}z\mid^{k}\leq \sqrt{2}C(x)\sum_{k=p_{0}}^{p}\frac{2^{\frac{k}{2}}(k+1)^{\frac{1}{4}}}{\sqrt{k!}}
\mid\sqrt{\frac{\nu}{2}}z\mid^{k}.
\end{align}
Applying the asymptotic Stirling formula \cite[p.312]{Nik}
\begin{align}\label{E2.81}
k!\sim \sqrt{2\pi
k}\bigg(\frac{k}{e}\bigg)^{k},\hspace{0.2cm}\mbox{as}\hspace{0.2cm}k\longrightarrow+\infty,
\end{align}
the inequality given in $(\ref{E2.80})$ can be rewritten as
\begin{align}\label{E2.82}
\nonumber\sum_{k=p_{0}}^{p}\frac{\mid H_{k+1}(x)\mid}{k!}
\mid\sqrt{\frac{\nu}{2}}z\mid^{k}&\leq C_{1}(x)\sum_{k=p_{0}}^{p}\frac{(k+1)^{\frac{1}{4}}}{k^{\frac{1}{4}}k^{\frac{k}{2}}}
\mid\sqrt{\nu}z\mid^{k}e^{\frac{k}{2}}\\
&=C_{1}(x)\sum_{k=p_{0}}^{p}\frac{(k+1)^{\frac{1}{4}}}{k^{\frac{1}{4}}}\frac{e^{\frac{k}{2}}\mid\sqrt{\nu}z\mid^{k}}{k^{\frac{k}{4}}}
\frac{1}{k^{\frac{k}{4}}}.
\end{align}
where $C_{1}(x)$ is a positive constant.\\
Considering  the fact that
$\frac{e^{\frac{k}{2}}\mid\sqrt{\nu}z\mid^{k}}{k^{\frac{k}{4}}} $
goes to zero when $k$ goes to infinity, we can obtain, from the
above inequality, the following estimate
\begin{align}\label{E2.83}
\nonumber \sum_{k=p_{0}}^{p}\frac{\mid H_{k+1}(x)\mid}{k!}
\mid\sqrt{\frac{\nu}{2}}z\mid^{k}&\leq C_{2}(x)\sum_{k=p_{0}}^{p}\frac{1}{k^{\frac{k}{4}}}\\
&\leq  C_{2}(x)\sum_{k=p_{0}}^{+\infty}\frac{1}{k^{\frac{k}{4}}}=:M(x)<+\infty,
\end{align}
where $C_{2}(x)$ is a positive constant and the convergence of the last series is assured by the Cauchy criterion.\\
Returning back to the inequality  $(\ref{E2.79})$ and using  $(\ref{E2.83})$, we get the following estimate
\begin{align}\label{E2.84}
 \mid\sum_{k=0}^{p}\frac{H_{k+1}(x)}{k!}
e^{-k t}(\sqrt{\frac{\nu}{2}}z)^{k}\mid\leq C_{\nu}(p_{0},x,z),\hspace{0.2cm}\forall\hspace{0.2cm} p\geq p_{0},
\end{align}
where $C_{\nu}(p_{0},x,z)=M(x)+\displaystyle{\sum_{k=0}^{p_{0}}}\frac{\mid H_{k+1}(x)\mid}{k!}\mid\sqrt{\frac{\nu}{2}}z\mid^{k}.$\\
This last estimate combined with the fact $\int_{0}^{+\infty}\sqrt{t}e^{-t}dt<+\infty$, help us to use the Lebesgue dominate convergence
 theorem, for interchanging  the sum and the integral given in $(\ref{E2.76})$. Then, we obtain the following equality
\begin{align}\label{E2.85}
 T_{\nu}(z,x)=\frac{1}{\Gamma(\frac{3}{2})}\int_{0}^{+\infty}[\sum_{k=0}^{+\infty}\frac{H_{k+1}(x)}{k!}
(e^{-t}\sqrt{\frac{\nu}{2}}z)^{k}]e^{-t}\sqrt{t}dt.
\end{align}
Next, applying the generating function \cite[p.102]{Mou}
\begin{align}\label{E2.86}
\sum_{k=0}^{+\infty}\frac{H_{k+\ell}(x)}{k!}s^{k}=\exp(2xs-s^{2})H_{\ell}(x-s),
\end{align}
for $\ell=1$ and $s=\sqrt{\frac{\nu}{2}}e^{-t}z$, we find
\begin{align}\label{E2.87}
T_{\nu}(z,x)=\frac{1}{\Gamma(\frac{3}{2})}\int_{0}^{+\infty}\sqrt{t}e^{-t} \exp(x\sqrt{2\nu}e^{-t}z-\frac{\nu}{2}
 e^{-2t}z^{2})H_{1}(x-\sqrt{\frac{\nu}{2}}e^{-t}z)dt.
\end{align}
By using the fact $\Gamma(\frac{3}{2})=\frac{\sqrt{\pi}}{2}$ (\cite[p.311]{Nik}), the above equation becomes
 \begin{align}\label{E2.88}
T_{\nu}(z,x)=\frac{2}{\sqrt{\pi}}\int_{0}^{+\infty}\sqrt{t}e^{-t} \exp(x\sqrt{2\nu}e^{-t}z-\frac{\nu}{2}
 e^{-2t}z^{2})H_{1}(x-\sqrt{\frac{\nu}{2}}e^{-t}z)dt.
\end{align}
Finally, returning back to $(\ref{E2.72})$ and replacing the function $T_{\nu}(z,x)$ by its expression given in $(\ref{E2.88})$,
 we get, for the kernel defined in $(\ref{E2.70})$, the following formula
\begin{align}\label{E2.89}
 \tilde{K}_{\nu}(z,x)=\sqrt{\nu}(\pi)^{\frac{-3}{4}}[1+\frac{\sqrt{2}}{\sqrt{\pi}}z\int_{0}^{+\infty}\sqrt{t}e^{-t}
\exp(x\sqrt{2\nu}e^{-t}z-\nu e^{-2t}\frac{z^{2}}{2})H_{1}(x-\sqrt{\frac{\nu}{2}}e^{-t}z)dt].
\end{align}
The aim of the second step is to prove that the integral transformation
\begin{align}\label{E2.90}
 \nonumber\hspace{0.25cm}\tilde{B}_{\nu} \hspace{0.2cm}:L^{2}(\mathbb{R},& \hspace{0.2cm}e^{-x^{2}}dx)\longrightarrow \mathcal{A}_{1}^{2,\nu}(\mathbb{C})\\
 &\varphi\longmapsto \tilde{B}_{\nu}[\varphi](z):=\int_{-\infty}^{+\infty}\tilde{K}_{\nu}(z,x)\varphi(x)e^{-x^{2}}dx
 \end{align}
is an isometry operator. For this, we will prove that the function $\tilde{K}_{\nu}(z,.)$ belongs to $L^{2}(\mathbb{R},\hspace{0.2cm}e^{-x^{2}}dx)$,
 for each fixed $z\in \mathbb{C}.$\\
Recall that the function $\tilde{K}_{\nu}(z,x)$ was defined in the formula $(\ref{E2.70})$ by
\begin{align}\label{E2.91}
 \tilde{K}_{\nu}(z,x)=\sum_{j=0}^{+\infty}\varphi_{j}(x)\psi_{j}(z).
\end{align}
Then, by the use of the Parseval's formula for the Fourier series in a Hilbert space, we have
\begin{align}\label{E2.92}
 \nonumber \parallel \tilde{K}_{\nu}(z,.)\parallel_{L^{2}(\mathbb{R},\hspace{0.2cm}e^{-x^{2}}dx)}^{2}&=\sum_{j=0}^{+\infty}\mid\psi_{j}(z)\mid^{2}\\
\nonumber&=\frac{\nu}{\pi}+\frac{1}{\pi}\sum_{j=1}^{+\infty}\frac{1}{j(j!)}\mid\nu z^{2}\mid^{j}\\
&\leq \frac{1}{\pi}[\nu+\exp(\nu \mid z\mid^{2})].
\end{align}
This proves that $\tilde{K}_{\nu}(z,.)\in L^{2}(\mathbb{R},\hspace{0.2cm}e^{-x^{2}}dx)$ for every fixed $z\in \mathbb{C}$. For the kernel function $\tilde{K}_{\nu}$, we define the integral transform $\tilde{B}_{\nu}$ by
\begin{align}\label{E2.93}
 \nonumber \tilde{B}_{\nu}[\varphi](z)&=<\tilde{K}_{\nu}(z,.),\overline{\varphi}>_{L^{2}(\mathbb{R},\hspace{0.2cm}e^{-x^{2}}dx)}\\
&=\int_{-\infty}^{+\infty}\tilde{K}_{\nu}(z,x)\varphi(x)e^{-x^{2}}dx,\hspace{0.25cm}z\in \mathbb{C},
\end{align}
provided that the integral exists.\\
Applying the Cauchy-Schwartz inequality to the first equality in the formula $(\ref{E2.93})$, we get
\begin{align}\label{E2.94}
\nonumber \mid\tilde{B}_{\nu}[\varphi](z)\mid&=\mid<\tilde{K}_{\nu}(z,.),\overline{\varphi}>_{L^{2}(\mathbb{R},\hspace{0.2cm}e^{-x^{2}}dx)}\mid\\
&\leq\parallel\tilde{K}_{\nu}(z,.)\parallel_{L^{2}(\mathbb{R},\hspace{0.2cm}e^{-x^{2}}dx)}
\parallel\varphi\parallel_{L^{2}(\mathbb{R},\hspace{0.2cm}e^{-x^{2}}dx)}.
\end{align}
From the inequality given in $(\ref{E2.92})$, one can obtain the following estimate
\begin{align}\label{E2.95}
\mid\tilde{B}_{\nu}[\varphi](z)\mid\leq \sqrt{\frac{1}{\pi}[\nu+\exp(\nu \mid z\mid^{2})]}\parallel\varphi\parallel_{L^{2}(\mathbb{R},\hspace{0.2cm}e^{-x^{2}}dx)},
\end{align}
for all $\varphi \in L^{2}(\mathbb{R},\hspace{0.2cm}e^{-x^{2}}dx)$
and $z$ fixed in $\mathbb{C}$. The inequality $(\ref{E2.95})$
traduces the continuity of the following linear functional
\begin{align}\label{E2.96}
\nonumber L^{2}(\mathbb{R},\hspace{0.2cm} &e^{-x^{2}}dx)\longrightarrow \mathbb{C}\\
 &\varphi\longmapsto \tilde{B}_{\nu}[\varphi](z).
 \end{align}
Note that it is quite easy to see that, for every nonnegative
integer $j\in \mathbb{Z}_{+}$, we have
\begin{align}\label{E2.97}
 \tilde{B}_{\nu}[\varphi_{j}](z)=\psi_{j}(z).
\end{align}
Indeed, this is an immediate consequence of the equations $(\ref{E2.91})$ and $(\ref{E2.93})$. The continuity of the linear functional $\varphi\longmapsto \tilde{B}_{\nu}[\varphi](z)$ ensures the following equality
\begin{align}\label{E2.98}
 \nonumber \tilde{B}_{\nu}[\varphi](z)&=\sum_{j=0}^{+\infty}\lambda_{j}\tilde{B}_{\nu}[\varphi_{j}](z)\\
&=\sum_{j=0}^{+\infty}\lambda_{j}\psi_{j}(z),
\end{align}
for every $\varphi=\displaystyle{\sum_{j=0}^{+\infty}}\lambda_{j}\varphi_{j}$ in $L^{2}(\mathbb{R},\hspace{0.2cm}e^{-x^{2}}dx)$.
 Moreover, $\tilde{B}_{\nu}[\varphi](z)$ converges absolutely for all $z\in \mathbb{C}$. Indeed, by Cauchy-Schwarz inequality,
  we obtain
\begin{align}\label{E2.99}
 \nonumber \mid \tilde{B}_{\nu}[\varphi](z)\mid&\leq \sum_{j=0}^{+\infty}\mid\lambda_{j}\mid \mid \psi_{j}(z)\mid\\
\nonumber &\leq  (\sum_{j=0}^{+\infty}\mid\lambda_{j}\mid^{2})^{\frac{1}{2}}(\sum_{j=0}^{+\infty}\mid \psi_{j}(z)\mid^{2})^{\frac{1}{2}}\\
&\leq \sqrt{\frac{\nu +e^{\nu \mid z\mid^{2}}}{\pi}}\parallel\varphi\parallel_{L^{2}(\mathbb{R},\hspace{0.2cm}e^{-x^{2}}dx)},
 \end{align}
where we have used $(\ref{E2.92})$. Moreover, we have
\begin{align}\label{E2.100}
 \nonumber \parallel \tilde{B}_{\nu}[\varphi]\parallel_{\nu,1}^{2}&=\sum_{j=0}^{+\infty} \mid \lambda_{j}\mid^{2}\parallel \psi_{j}\parallel_{\nu,1}^{2}\\
 \nonumber &=\sum_{j=0}^{+\infty} \mid \lambda_{j}\mid^{2}\\
&=\parallel\varphi\parallel^{2}_{L^{2}(\mathbb{R},\hspace{0.2cm}e^{-x^{2}}dx)}.
 \end{align}
This shows that $\tilde{B}_{\nu}$ is a well defined isometry from $L^{2}(\mathbb{R},\hspace{0.2cm}e^{-x^{2}}dx)$ onto $\mathcal{A}_{1}^{2,\nu}(\mathbb{C})$.
Finally, by considering the following isometry

\begin{align}\label{E2.101}
 \nonumber B_{\nu}\hspace{0.2cm}:L^{2}(\mathbb{R},&\hspace{0.2cm} dx)\longrightarrow \mathcal{A}_{1}^{2,\nu}(\mathbb{C})\\
 &\varphi\longmapsto B_{\nu}[\varphi]=\tilde{B}_{\nu}\circ T[\varphi],
 \end{align}
where $T$ is the canonical isometry given by
\begin{align}\label{E2.102}
 \nonumber T\hspace{0.2cm}:L^{2}(\mathbb{R},&\hspace{0.2cm} dx)\longrightarrow L^{2}(\mathbb{R},\hspace{0.2cm} e^{-x^{2}}dx)\\
 &\varphi\longmapsto e^{\frac{x^{2}}{2}}\varphi,
 \end{align}
we give the desired result and then the proof of theorem $(\ref{T2.1})$ is complete.
\end{proof}
\section{Generalized Bargmann-Dirichlet spaces}
In this section,  we would like  to associate a new class of
Bargmann transforms to a class of generalized Bargmann-Dirichlet
 spaces $\mathcal{A}_{m}^{2,\nu}(\mathbb{C})$, $\nu>0$ and $m\in \mathbb{Z}_{+}$, called weighted Bargmann-Dirichlet
  spaces of order $m$. The functional spaces $\mathcal{A}_{m}^{2,\nu}(\mathbb{C})$ have been considered by Intissar
   \texttt{et} al. in \cite{Elh}. In order to avoid any confusion, it should be noted that the spaces $\mathcal{A}_{m}^{2,\nu}(\mathbb{C})$
   have been  noted in the last reference by $\mathbb{B}_{m}^{2,\nu}(\mathbb{C})$.
\subsection{Generalized Bargmann-Dirichlet spaces from magnetic Laplacian}
To introduce the weighted Bargmann-Dirichlet space $\mathcal{A}_{m}^{2,\nu}(\mathbb{C})$ of order $m$, we need to fixe some notations.\\
Let $\nu>0$ and $L^{2,\nu}(\mathbb{C})=L^{2,\nu}(\mathbb{C},\hspace{0.2cm}e^{-\nu \mid z\mid^{2}}d\lambda(z))$ is the Hilbert space of all
square-integrable functions on $\mathbb{C}$ with respect to the Gaussian measure $d\mu_{\nu}(z)=e^{-\nu \mid z\mid^{2}}d\lambda(z)$,
 where $d\lambda(z)$ is the Lebesgue measure. The hermitian scalar product is defined by
\begin{align}\label{E3.1}
<f,g>_{\nu}=\int_{\mathbb{C}}f(z)\overline{g(z)}e^{-\nu\mid z\mid^{2}}d\lambda(z).
 \end{align}
The associated norm is given by
\begin{align}\label{E3.2}
\parallel f\parallel_{\nu}^{2}=\int_{\mathbb{C}}\mid f(z) \mid^{2}e^{-\nu\mid z\mid^{2}}d\lambda(z).
 \end{align}
In the analogue way of the Bargmann-Dirichlet space $\mathcal{A}_{1}^{2,\nu}(\mathbb{C})$ dealt with in the last section,
we consider here the following functional space $\mathcal{A}_{m}^{2,\nu}(\mathbb{C})$ defined by
\begin{align}\label{E3.3}
\mathcal{A}_{m}^{2,\nu}(\mathbb{C})=\{f:\mathbb{C}\longrightarrow\mathbb{C}, \hspace{0.15cm} \mbox{holomorphic and}
 \hspace{0.15cm}\int_{\mathbb{C}}\mid f^{(m)}(z)\mid^{2}e^{-\nu\mid z\mid^{2}}d\lambda(z)<+\infty\},
\end{align}
where $f^{(m)}(z):=\frac{\partial^{m} f(z)}{\partial z^{m}}$ indicates the complex $m-$derivative of the function $f$.\\
For fixed nonnegative integer $m=2,3,...$, any holomorphic function
$f(z)=\displaystyle{\sum_{k=0}^{+\infty}}a_{k}z^{k}$ on $\mathbb{C}$
can be written as
\begin{align}\label{E3.4}
f(z)=f_{1,m}(z)+f_{2,m}(z),
\end{align}
where  $f_{1,m}(z)=\displaystyle{\sum_{k=0}^{m-1}}a_{k}z^{k}$ and
$f_{2,m}(z)=\displaystyle{\sum_{k=m}^{+\infty}}a_{k}z^{k}$.  The
space $\mathcal{A}_{m}^{2,\nu}(\mathbb{C})$ can be equipped with the
following norm
\begin{align}\label{E3.5}
\parallel f\parallel_{\nu,m}^{2}=\parallel f_{1,m}\parallel_{\nu}^{2}+\parallel f_{2,m}^{(m)}\parallel_{\nu}^{2}.
\end{align}
It is noted that the hermitian inner product $<,>_{\nu,m}$ associated with the norm\\
 $\parallel. \parallel_{\nu,m}$ is given through
\begin{align}\label{E3.6}
<f,g>_{\nu,m}=<f_{1,m},g_{1,m}>_{\nu}+<f_{2,m}^{(m)},g_{2,m}^{(m)}>_{\nu}.
\end{align}
Note that the monomials $e_{j}(z)=z^{j}$, $j\in \mathbb{Z}_{+}$
belong to $\mathcal{A}_{m}^{2,\nu}(\mathbb{C})$ \cite{Elh},
 and they are pairwise orthogonal with respect to the hermitian scalar product $<,>_{\nu,m}$ with

\begin{align}\label{E3.7}
\parallel e_{j}\parallel^{2}_{\nu,m}=
\left\{
  \begin{array}{ll}
    \frac{j!\pi}{\nu^{j+1}}, &\mbox{for} \hspace{0.25cm}j\leq m-1, \hbox{} \\\\
    \frac{(j!)^{2}\pi}{\nu^{j-m+1}\Gamma(j-m+1)}, &\mbox{for} \hspace{0.25cm}j\geq m. \hbox{}
  \end{array}
\right.
\end{align}
Then, an orthonormal basis for the space $\mathcal{A}_{m}^{2,\nu}(\mathbb{C})$ can be given by:
\begin{align}\label{E3.8}
   \psi_{j}^{m}(z)=\left\{
                 \begin{array}{ll}
                   \frac{\nu^{\frac{j+1}{2}}}{\sqrt{\pi j!}}z^{j},&j\leq m-1\hbox{} \\\\
                   \frac{(\nu)^{\frac{j-m+1}{2}}\sqrt{\Gamma(j-m+1)}}{\sqrt{\pi}j!}z^{j},&j\geq m. \hbox{}
                 \end{array}
               \right.
\end{align}
According to \cite{Elh}, it is proved that the generalized Bargmann-Dirichlet space $\mathcal{A}_{m}^{2,\nu}(\mathbb{C})$
 is a reproducing kernel Hilbert space. Concretely, the reproducing kernel for the space $\mathcal{A}_{m}^{2,\nu}(\mathbb{C})$
  has been given explicitly by the following formula
\begin{align}\label{E3.9}
\tilde{K}_{m,\nu}(z,w)=\frac{\pi}{\nu}\{\sum_{j=0}^{m-1}\frac{(\nu z \overline{w})^{j}}{j!}+\frac{(z\overline{w})^{m}}
{(m!)^{2}}\prescript{}{2}{F}_2^{}(1,1;m+1,m+1;\nu z \overline{w})\},
\end{align}
where $\prescript{}{2}{F}_2^{}$ is the hypergeometric function defined in $(\ref{E2.10})$.\\
For a holomorphic function $f(z)=\displaystyle{\sum_{j=0}^{+\infty}}a_{j}z^{j}$, we have the following membership test
\begin{align}\label{E3.10}
f(z):=\sum_{j=0}^{+\infty} a_{j}z^{j}\in \mathcal{A}_{m}^{2,\nu}(\mathbb{C})\Longleftrightarrow
\sum_{j=m}^{+\infty}\frac{(j!)^{2}}{\nu^{j-m} (j-m)!} \mid a_{j} \mid^{2}<+\infty.
\end{align}
As in the case of the space  $\mathcal{A}_{1}^{2,\nu}(\mathbb{C})$, we relate the generalized Bargmann-Dirichlet space
 $\mathcal{A}_{m}^{2,\nu}(\mathbb{C})$ to the magnetic Laplacian on the complex plane. Precisely, we have the following proposition.
\begin{proposition}\label{P3.1}
Let  $\nu>0$, $m\in \mathbb{Z}_{+}$, $m\geq2$  and let $\Delta_{\nu}$ be the partial differential operator defined by
\begin{align}\label{E3.11}
\Delta_{\nu}=-\frac{\partial^{2}}{\partial z\partial\overline{z}}+\nu\overline{z}\frac{\partial}{\partial\overline{z}},
\end{align}
acting on the Hilbert space  $L^{2,\nu}(\mathbb{C}):=L^{2}(\mathbb{C},\hspace{0.2cm}e^{-\nu\mid z\mid^{2}}d\lambda(z))$ with the dense domain
\begin{align}\label{E3.12}
D_{m}(\Delta_{\nu})=\{F\in L^{2,\nu}(\mathbb{C}),\hspace{0.25cm}\Delta_{\nu}F\in L^{2,\nu}(\mathbb{C})\hspace{0.25cm}
\mbox{and}\hspace{0.25cm}\frac{\partial^{m} F}{\partial z^{m} }\in L^{2,\nu}(\mathbb{C})\}.
\end{align}
Then, we have
\begin{align}\label{E3.13}
\mathcal{A}^{2,\nu}_{m}(\mathbb{C})=\{F\in D_{m}(\Delta_{\nu}),\hspace{0.25cm}\Delta_{\nu}F=0\}.
\end{align}
\end{proposition}
\begin{proof}
Let $F\in D_{m}(\Delta_{\nu})$, such that $\Delta_{\nu} F=0$. Then, we have
\begin{center}
$\textbf{(i)}$  $F\in L^{2,\nu}(\mathbb{C}),$ $\textbf{(ii)}$  $\frac{\partial^{m} F}{\partial z^{m}}\in L^{2,\nu}(\mathbb{C})$ and $\textbf{(iii)}$
 $\Delta_{\nu}F=0.$
\end{center}
By the proposition $(\ref{P2.3})$, we see that the both conditions
$\textbf{(i)}$ and $\textbf{(ii)}$ imply  that $F$ is holomorphic.
Combining this property with the fact that $\frac{\partial^{m}
F}{\partial z^{m}}\in L^{2,\nu}(\mathbb{C})$,
we show that $F$ belongs to the generalized Bargmann-Dirichlet space $\mathcal{A}_{m}^{2,\nu}(\mathbb{C})$.\\
Conversely, let  $F(z)=\displaystyle{\sum_{k=0}^{+\infty}}a_{k}z^{k}$ be a holomorphic function on $\mathbb{C}$ such
that\\ $\int_{\mathbb{C}}\mid F^{(m)}(z)\mid^{2}e^{-\nu\mid z\mid^{2}}d\lambda(z)<+\infty$. By the holomorphicity of $F$, we
 have $\frac{\partial F}{\partial\overline{z}}=0$, then we obtain
\begin{align}\label{E3.14}
-\frac{\partial^{2} F}{\partial z\partial\overline{z}}+\nu\overline{z}\frac{\partial F}{\partial\overline{z}}=0.
\end{align}
For proving that $F$ is in the kernel of the operator $\Delta_{\nu}$
defined by $(\ref{E3.11})$ and $(\ref{E3.12})$, we should
 to show that $F\in L^{2,\nu}(\mathbb{C})$. It is not hard to see that we have the following inequality
\begin{align}\label{E3.15}
\sum_{j=m}^{+\infty} \frac{1}{\nu^{j+1}}j!\mid a_{j}\mid^{2}\leq
\frac{1}{\nu^{m+1}}\sum_{j=m}^{+\infty}  \frac{(j!)^{2}}{\nu^{j-m}(j-m)!}\mid a_{j}\mid^{2}.
\end{align}
Using the membership test $(\ref{E3.10})$ with the fact
\begin{align}\label{E3.16}
\int_{\mathbb{C}}\mid F(z)\mid^{2}e^{-\nu\mid z\mid^{2}}d\lambda(z)=\frac{\pi}{\nu}\sum_{j=0}^{+\infty} \frac{j!}{\nu^{j}}\mid a_{j}\mid^{2},
\end{align}
given in $(\ref{E2.55})$, we obtain from the above inequality that $F$ belongs to the Hilbert space $L^{2,\nu}(\mathbb{C})$.
 Thus, one has  $F\in \{F\in D_{m}(\Delta_{\nu}),\hspace{0.25cm}\Delta_{\nu}F=0\}$. Hence, we have proved that
\begin{align}\label{E3.17}
\mathcal{A}^{2,\nu}_{m}(\mathbb{C})=\{F\in D_{m}(\Delta_{\nu}),\hspace{0.25cm}\Delta_{\nu}F=0\}.
\end{align}
\end{proof}
In the same way as in the proposition $(\ref{P2.5})$, we can state the following result.
\begin{proposition}\label{P3.2} We have the following properties:
\begin{enumerate}
 \item The operator $\Delta_{\nu}$ acting in $D_{m}(\Delta_{\nu})$ is closable and admits a self-adjoint extension.
  \item The operator $\Delta_{\nu}$ with the domain $D_{m}(\Delta_{\nu})$ is an unbounded non self-adjoint operator.
  \item $0$ belongs to the point spectrum of $\Delta_{\nu}$ considered on the domain $D_{m}(\Delta_{\nu})$.
\end{enumerate}
\end{proposition}
\begin{proof}
The proof is the same as for the proposition $(\ref{P2.5})$ and can be omitted.
\end{proof}
\subsection{Bargmann transform associated with generalized Bargmann-Dirichlet spaces}
To the generalized Bargmann-Dirichlet space $\mathcal{A}_{m}^{2,\nu}(\mathbb{C})$, we shall associate a Bargmann transform.
 Precisely, we have the following theorem
\begin{theorem}\label{T3.1} Let $\nu$ and $m\in \mathbb{Z}_{+}$, $m\geq2$. Then, we have the following unitary isomorphism
\begin{align}\label{E3.18}
 \nonumber B_{\nu,m}\hspace{0.2cm}:L^{2}(\mathbb{R},&\hspace{0.2cm} dx)\longrightarrow \mathcal{A}_{m}^{2,\nu}(\mathbb{C})\\
 &\varphi\longmapsto B_{\nu,m}[\varphi](z):=\int_{-\infty}^{+\infty}K_{\nu,m}(z,x)\varphi(x)dx.
 \end{align}
 The integral kernel $K_{\nu,m}(z,x)$ is given by
 \begin{align}\label{E3.19}
\nonumber K_{\nu,m}(z,x)&=\sqrt{\nu}(\pi)^{\frac{-3}{4}}e^{\frac{-x^{2}}{2}}[\sum_{k=0}^{m-1}(\sqrt{\frac{\nu}{2}}z)^{k}\frac{H_{k}(x)}{k!}\\
&+\big(\sqrt{\frac{2}{\pi}}\big)^{m}z^{m}\int_{0}^{+\infty}\varpi_{m}(t)
\exp(x\sqrt{2\nu}e^{-t}z-\nu e^{-2t}\frac{z^{2}}{2})H_{m}(x-\sqrt{\frac{\nu}{2}}e^{-t}z)dt],
\end{align}
where $H_{m}(x)$ is the Hermite polynomials and $\varpi_{m}(t)$ is the function defined by
\begin{align}\label{E3.20}
\varpi_{m}(t)=(\sqrt{t}e^{-t})*(\sqrt{t}e^{-2t})*...*(\sqrt{t}e^{-mt}), \hspace{0.25cm}m\geq2.
\end{align}
The notation $f*g$ means the following convolution product \cite[p.91]{Sch}
\begin{align}\label{E3.21}
f*g(x)=\int_{0}^{x}f(x-y)g(y)dy.
\end{align}
\end{theorem}
Before giving the proof of the theorem, we will state the following more precise result for $m=2$. Concretely, we have the following proposition
\begin{proposition}\label{P3.12} Let $\nu>0$ and $m=2$. Then, we have the following unitary isomorphism
\begin{align}\label{E3.22}
 \nonumber B_{\nu,2}\hspace{0.2cm}:L^{2}(\mathbb{R},& \hspace{0.2cm}dx)\longrightarrow \mathcal{A}_{2}^{2,\nu}(\mathbb{C})\\
 &\varphi\longmapsto B_{\nu,2}[\varphi](z):=\int_{-\infty}^{+\infty}K_{\nu,2}(z,x)\varphi(x)dx,
 \end{align}
  where the integral kernel $K_{\nu,2}(z,x)$ is given by
 \begin{align}\label{E3.23}
\nonumber K_{\nu,2}(z,x)&=\sqrt{\nu}(\pi)^{\frac{-3}{4}}e^{\frac{-x^{2}}{2}}[1+\sqrt{2\nu}xz\\
&+\frac{z^{2}}{4}\int_{0}^{+\infty}t^{2}
\exp(-2t+x\sqrt{2\nu}e^{-t}z-\nu e^{-2t}\frac{z^{2}}{2})\prescript{}{1}{F}_1^{}(\frac{3}{2};3;t)H_{2}(x-\sqrt{\frac{\nu}{2}}e^{-t}z)dt].
\end{align}
$\prescript{}{1}{F}_1^{}(\alpha;\beta;t)=\displaystyle{\sum_{j=0}^{+\infty}}\frac{(\alpha)_{j}}{(\beta)_{j}}\frac{z^{j}}{j!}$
is the confluent hypergeometric function and $H_{2}(y)=4y^{2}-2$.
\end{proposition}
In order to prove the above theorem, we need some auxiliary results. Precisely, with the help of the notations given in theorem
 $(\ref{T3.1})$, we have the following lemma.

\begin{lemma} \label{L3.1}Let $m\in\mathbb{Z}^{+}$ such that $m\geq2$.
\begin{enumerate}
  \item The following estimate holds
  \begin{align}\label{E3.24}
  \varpi_{m}(t)\leq ({\cal{B}}(\frac{3}{2},\frac{3}{2}))^{m-1}(\sqrt{t})^{3m-2}e^{-t},
  \end{align}
  where ${\cal{B}}(x,y)$ denotes the beta special function \cite[p.7]{Mag}, define by
  \begin{align}
  {\cal{B}}(x,y)=\int_{0}^{1}t^{x-1}(1-t)^{y-1}dt, \hspace{0.2cm}x>0,\hspace{0.2cm}y>0.
  \end{align}
  \item The Laplace transform of $\varpi_{(\alpha,m)}(t)$ is well defined. Moreover, we have
\begin{align}\label{E3.25}
\mathscr{L}(\varpi_{m}(t))(k)=\frac{(\Gamma(\frac{3}{2}))^{m}}{[(k+1)(k+2)...(k+m)]^{\frac{3}{2}}},\hspace{0.25cm}
k \in \mathbb{Z}_{+},
\end{align}
where $\mathscr{L}$ denotes the classical Laplace transform defined by \cite[p.2]{Sch}
\begin{align}\label{E3.26}
\mathscr{L}(f(t))(\lambda):=\int_{0}^{+\infty}e^{-\lambda t}f(t)dt, \hspace{0.25cm}\lambda>0.
\end{align}
\end{enumerate}
\end{lemma}
\begin{proof}
 For $m=2$, we obtain
\begin{align}\label{E3.27}
\nonumber\varpi_{2}(t)&=(\sqrt{t}e^{-t})*(\sqrt{t}e^{-2t})\\
\nonumber&=\int_{0}^{t}\sqrt{t-s}e^{-(t-s)}\sqrt{s}e^{-2s}ds\\
\nonumber&=e^{-t}\int_{0}^{t}\sqrt{t-s}\sqrt{s}e^{-s}ds\\
\nonumber&\leq e^{-t}\int_{0}^{t}\sqrt{t-s}\sqrt{s}ds\\
\nonumber&=e^{-t}\int_{0}^{1}t\sqrt{t-tx}\sqrt{tx}dx, \hspace{0.25cm} (x:=\frac{s}{t})\\
\nonumber&=e^{-t}t^{2}\int_{0}^{1}\sqrt{1-x}\sqrt{x}dx\\
\nonumber&=e^{-t}t^{2}{\cal{B}}(\frac{3}{2},\frac{3}{2})\\
&={\cal{B}}(\frac{3}{2},\frac{3}{2})(\sqrt{t})^{3\times2-2}e^{-t}.
\end{align}
For $m=3$, however,  we get the following estimate
\begin{align}\label{E3.28}
\nonumber\varpi_{3}(t)&=\varpi_{2}(t)*(\sqrt{t}e^{-3t})\\
\nonumber&\leq{\cal{B}}(\frac{3}{2},\frac{3}{2})\int_{0}^{t}\sqrt{t-s}e^{-3(t-s)})(\sqrt{s})^{3\times2-2}e^{-s}ds\\
\nonumber&\leq{\cal{B}}(\frac{3}{2},\frac{3}{2})e^{-t}(\sqrt{t})^{3}\int_{0}^{t}\sqrt{t-s}\sqrt{s}ds\\
\nonumber&\leq{\cal{B}}(\frac{3}{2},\frac{3}{2})e^{-t}(\sqrt{t})^{3}({\cal{B}}(\frac{3}{2},\frac{3}{2})t^{2})\\
         &\leq({\cal{B}}(\frac{3}{2},\frac{3}{2}))^{2}(\sqrt{t})^{3\times3-2}e^{-t}.
\end{align}
Step by step, we obtain that
\begin{align}\label{E3.29}
\nonumber\varpi_{m}(t)&=\varpi_{m-1}(t)*(\sqrt{t}e^{-mt})\\
\nonumber&=\int_{0}^{t}\varpi_{m-1}(s)\sqrt{t-s}e^{-m(t-s)}ds\\
\nonumber&\leq ({\cal{B}}(\frac{3}{2},\frac{3}{2}))^{m-2}\int_{0}^{t}(\sqrt{t})^{3(m-1)-2}e^{-s}\sqrt{t-s}e^{-m(t-s)}ds\\
\nonumber&\leq ({\cal{B}}(\frac{3}{2},\frac{3}{2}))^{m-2}
(\sqrt{t})^{3(m-1)-3}e^{-t}\int_{0}^{t}\sqrt{t-s}\sqrt{s}ds\\
\nonumber&= ({\cal{B}}(\frac{3}{2},\frac{3}{2}))^{m-2}
(\sqrt{t})^{3(m-1)-3}e^{-t}({\cal{B}}(\frac{3}{2},\frac{3}{2})t^{2})\\
\nonumber&= ({\cal{B}}(\frac{3}{2},\frac{3}{2}))^{m-2}
(\sqrt{t})^{3m-6}e^{-t}{\cal{B}}(\frac{3}{2},\frac{3}{2})(\sqrt{t})^{4}\\
&=({\cal{B}}(\frac{3}{2},\frac{3}{2}))^{m-1}(\sqrt{t})^{3m-2}e^{-t}.
\end{align}
By the induction principle, we confirm the following inequality
\begin{align}\label{E3.30}
\varpi_{m}(t)\leq ({\cal{B}}(\frac{3}{2},\frac{3}{2}))^{m-1}(\sqrt{t})^{3m-2}e^{-t},\hspace{0.2cm}m\geq2.
\end{align}
The above inequality proves that the Laplace transform of
$\varpi_{m}(t)$ is well defined. Using the formula \cite[p.92]{Sch}
\begin{align}\label{E3.31}
\mathscr{L}(f*g)=\mathscr{L}(f)\mathscr{L}(g),
\end{align}
and  the  relation \cite[p.28]{Pru}
\begin{align}\label{E3.32}
\mathscr{L}(t^{a}e^{-bt})(k)=\frac{\Gamma(a+1)}{(k+b)^{a+1}},\hspace{0.2cm} a>-1,\hspace{0.2cm} b>0,
\end{align}
we get  the following required equality
\begin{align}\label{E3.33}
\nonumber \mathscr{L}(\varpi_{m}(t))(k)&=\mathscr{L}(\sqrt{t}e^{-t}*\sqrt{t}e^{-2t}*...*\sqrt{t}e^{-mt})(k)\\
&=\frac{(\Gamma(\frac{3}{2}))^{m}}{[(k+1)(k+2)...(k+m)]^{\frac{3}{2}}}.
\end{align}
\end{proof}
Now, we are in position to prove the theorem $(\ref{T3.1})$.
\begin{proof}
We will proceed as in theorem $(\ref{T2.1})$. To do so, we begin by considering the following kernel
\begin{align}\label{E3.34}
 \tilde{K}_{\nu,m}(z,x)=\sum_{j=0}^{+\infty}\varphi_{j}(x)\psi_{j}^{m}(z),
\end{align}
associated with the orthonormal basis $\{\psi^m_{j}\}_{i\in\mathbb{Z}_{+}}$, appearing in  $(\ref{E3.8})$ and
the orthonormal basis of $L^{2}(\mathbb{R},\hspace{0.2cm}e^{-x^{2}}dx)$ given in $(\ref{E2.71})$.\\
First, we have to calculate the kernel $ K_{\nu,m}(z,x)$. Indeed, we express it in terms of the explicit forms of
 $\varphi_{j}(x)$ and $\psi_{j}^{m}(z)$ as follows
\begin{align}\label{E3.35}
 \nonumber \tilde{K}_{\nu,m}(z,x)&=\pi^{\frac{-3}{4}}[\sum_{j=0}^{m-1}2^{\frac{-j}{2}}\nu^{\frac{j+1}{2}}\frac{z^{j}}{j!}H_{j}(x)
+\sum_{j=m}^{+\infty}\nu^{\frac{j-m+1}{2}}2^{-\frac{j}{2}}\frac{\sqrt{\Gamma(j-m+1)}}{j!\sqrt{j!}}z^{j}H_{j}(x)]\\
\nonumber&=\pi^{\frac{-3}{4}}[\sum_{k=0}^{m-1}2^{\frac{-k}{2}}\nu^{\frac{k+1}{2}}\frac{z^{k}}{k!}H_{k}(x)
+\sum_{k=0}^{+\infty}\nu^{\frac{k+1}{2}}2^{-\frac{k+m}{2}}\frac{\sqrt{\Gamma(k+1)}k!}{[(k+m)!]^{\frac{3}{2}}}
\frac{z^{k+m}}{k!}H_{k+m}(x)]\\
&=\pi^{\frac{-3}{4}}\sqrt{\nu}[\sum_{k=0}^{m-1}\frac{(\sqrt{\frac{\nu}{2}}z)^{k}}{k!}H_{k}(x)
+2^{\frac{-m}{2}}z^{m}\sum_{k=0}^{+\infty}\frac{1}{[(k+1)(k+2)...(k+m)]^{\frac{3}{2}}}
\frac{(\sqrt{\frac{\nu}{2}}z)^{k}}{k!}H_{k+m}(x)].
\end{align}
As mentioned in the last section, the series
\begin{align}\label{E3.36}
 \sum_{k=0}^{+\infty}\frac{1}{[(k+1)(k+2)...(k+m)]^{\frac{3}{2}}} \frac{(\sqrt{\frac{\nu}{2}}z)^{k}}{k!}H_{k+m}(x)
\end{align}
involved in the right hand side of the equation $(\ref{E3.35})$ does
not appear in the literature as a standard closed
 generating formula for Hermite polynomials. To overcome  this problem, we use the point $(2)$ of the lemma $(\ref{L3.1})$.
  Then, the  equality $(\ref{E3.35})$ can be rewritten as
\begin{align}\label{E3.37}
 \tilde{K}_{\nu,m}(z,x)=\pi^{\frac{-3}{4}}\sqrt{\nu}[\sum_{k=0}^{m-1}\frac{(\sqrt{\frac{\nu}{2}}z)^{k}}{k!}H_{k}(x)
+2^{\frac{-m}{2}}(\Gamma(\frac{3}{2}))^{-m}z^{m}\sum_{k=0}^{+\infty}\int_{0}^{+\infty}e^{-kt}\varpi_{m}(t)
\frac{(\sqrt{\frac{\nu}{2}}z)^{k}}{k!}H_{k+m}(x)dt].
\end{align}
For given a closed formula of the kernel $\tilde{K}_{\nu,m}(z,x)$, we need to permute the integral and sum in the infinite series
\begin{align}\label{E3.38}
T_{\nu,m}(z,x)=\sum_{k=0}^{+\infty}\int_{0}^{+\infty}e^{-kt}\varpi_{m}(t)
\frac{(\sqrt{\frac{\nu}{2}}z)^{k}}{k!}H_{k+m}(x)dt,
 \end{align}
involved in the right hand side of the equality $(\ref{E3.37})$. Indeed, we use the formula $(\ref{E2.78})$ to obtain
\begin{align}\label{E3.39}
 \mid H_{k+m}(x)\mid\leq C(x)\frac{2^{\frac{k+m}{2}}((k+m)!)^{\frac{1}{2}}}{(k+m)^{\frac{1}{4}}},\hspace{0.2cm}
  \mbox{for} \hspace{0.2cm}m \hspace{0.2cm} \mbox{fixed and}\hspace{0.2cm} k \hspace{0.2cm} \mbox{enough large}.
\end{align}
Let $p_{0}$ be a fixed integer enough large and let $p\geq p_{0}$. Then, we have the following inequality
\begin{align}\label{E3.40}
 \mid\sum_{k=0}^{p}\frac{H_{k+m}(x)}{k!}
e^{-k t}(\sqrt{\frac{\nu}{2}}z)^{k}\mid\leq
\sum_{k=0}^{p_{0}}\frac{\mid H_{k+m}(x)\mid}{k!}
\mid\sqrt{\frac{\nu}{2}}z\mid^{k}
+\sum_{k=p_{0}}^{p}\frac{\mid H_{k+m}(x)\mid}{k!}
\mid\sqrt{\frac{\nu}{2}}z\mid^{k}, \hspace{0.25cm}t\geq0.
\end{align}
By exploiting the inequality $(\ref{E3.39})$, we obtain for the last sum in the right hand side of $(\ref{E3.40})$
\begin{align}\label{E3.41}
\nonumber \sum_{k=p_{0}}^{p}\frac{\mid H_{k+m}(x)\mid}{k!}
\mid\sqrt{\frac{\nu}{2}}z\mid^{k}&\leq C(x)\sum_{k=p_{0}}^{p}\frac{2^{\frac{k+m}{2}}[(k+m)!]^{\frac{1}{2}}}
{k!(k+m)^{\frac{1}{4}}}\mid\sqrt{\frac{\nu}{2}}z\mid^{k}\\
&\leq 2^{\frac{m}{2}} C(x)\sum_{k=p_{0}}^{p}\sqrt{\frac{(k+m)!}{k!}}\frac{1}{\sqrt{k!}}\frac{1}{(k+m)^{\frac{1}{4}}}
\mid\sqrt{\nu}z\mid^{k}.
\end{align}
By using the fact $j!=\Gamma(j+1)$, $j\geq0$, the above inequality can be rewritten as
\begin{align}\label{E3.42}
\sum_{k=p_{0}}^{p}\frac{\mid H_{k+m}(x)\mid}{k!}
\mid\sqrt{\frac{\nu}{2}}z\mid^{k}\leq 2^{\frac{m}{2}} C(x)\sum_{k=p_{0}}^{p}\sqrt{\frac{\Gamma(k+m+1)}
{\Gamma(k+1)}}\frac{1}{\sqrt{k!}}\frac{1}{(k+m)^{\frac{1}{4}}}
\mid\sqrt{\nu}z\mid^{k}.
\end{align}
With the help of the following asymptotic formula \cite[p.22]{Bea}
\begin{align}\label{E3.43}
\frac{\Gamma(y+a)}{\Gamma(y)}=y^{a}(1+O(y^{-1})),\hspace{0.25cm}\mbox{as} \hspace{0.25cm}y\longrightarrow+\infty,
\end{align}
applied to $y=k+1$ and $a=m$, we obtain the following estimate
\begin{align}\label{E3.44}
\sqrt{\frac{\Gamma(k+m+1)}{\Gamma(k+1)}}\leq C_{m} (1+k)^{\frac{m}{2}}, \hspace{0.2cm}\mbox{for} \hspace{0.2cm} k \hspace{0.2cm}
\mbox{enought large},
\end{align}
where $C_{m}$ is a positive constant. This last inequality combined with the asymptotic Stirling formula mentioned in $(\ref{E2.81})$
 help us to obtain from $(\ref{E3.42})$ the inequality

\begin{align}\label{E3.45}
\nonumber \sum_{k=p_{0}}^{p}\frac{\mid H_{k+m}(x)\mid}{k!}
\mid\sqrt{\frac{\nu}{2}}z\mid^{k}&\leq  C_{m}C(x)2^{\frac{m}{2}}\sum_{k=p_{0}}^{p}(1+k)^{\frac{m}{2}}\frac{1}{\sqrt{k!}}\frac{1}{(k+m)^{\frac{1}{4}}}
\mid\sqrt{\nu}z\mid^{k}\\
\nonumber & \leq\tilde{C}_{m}(x)\sum_{k=p_{0}}^{p} \frac{(1+k)^{\frac{m}{2}}}{(k+m)^{\frac{1}{4}}}k^{\frac{-1}{4}}(\frac{k}{e})^{\frac{-k}{2}}
\mid\sqrt{\nu}z\mid^{k}\\
&\leq \tilde{C}_{m}(x)\sum_{k=p_{0}}^{+\infty} \frac{(1+k)^{\frac{m}{2}}}{(k+m)^{\frac{1}{4}}}k^{\frac{-1}{4}}(\frac{k}{e})^{\frac{-k}{2}}\mid\sqrt{\nu}z\mid^{k}.
\end{align}
It is worth noting that the convergence of the last series is assured by the Cauchy convergence criterion. Thanks to the point $(1)$ of the lemma $(\ref{L3.1})$,
 we have the integrability of the function $\varpi_{m}(t)$ over the set $(0,+\infty)$. Then, by using the Lebesgue dominate convergence theorem, we can interchange the sum and the integral in $(\ref{E3.38})$. Then, the latter  becomes
\begin{align}\label{E3.46}
T_{\nu,m}(z,x)=\int_{0}^{+\infty}\sum_{k=0}^{+\infty}e^{-kt}\varpi_{m}(t)
\frac{(\sqrt{\frac{\nu}{2}}z)^{k}}{k!}H_{k+m}(x)dt.
 \end{align}
Applying the generating function mentioned in  $(\ref{E2.86})$
\begin{align}\label{E3.47}
\sum_{k=0}^{+\infty}\frac{H_{k+\ell}(x)}{k!}s^{k}=\exp(2xs-s^{2})H_{\ell}(x-s),
\end{align}
for $\ell=m$ and $s=\sqrt{\frac{\nu}{2}}e^{-t}z$, the equation $(\ref{E3.46})$ becomes
\begin{align}\label{E3.48}
T_{\nu,m}(z,x)=\int_{0}^{+\infty}\varpi_{m}(t)
\exp(x\sqrt{2\nu}e^{-t}z-\nu e^{-2t}\frac{z^{2}}{2})H_{m}(x-\sqrt{\frac{\nu}{2}}e^{-t}z)dt.
 \end{align}
Finally, the equation $(\ref{E3.37})$ takes the form
\begin{align}\label{E3.49}
\nonumber \tilde{K}_{\nu,m}(z,x)&=\pi^{\frac{-3}{4}}\sqrt{\nu}[\sum_{k=0}^{m-1}\frac{(\sqrt{\frac{\nu}{2}}z)^{k}}{k!}H_{k}(x)\\
&+2^{\frac{-m}{2}}(\Gamma(\frac{3}{2}))^{-m}z^{m}\int_{0}^{+\infty}\varpi_{m}(t)
\exp(x\sqrt{2\nu}e^{-t}z-\nu e^{-2t}\frac{z^{2}}{2})H_{m}(x-\sqrt{\frac{\nu}{2}}e^{-t}z)dt].
\end{align}
Now, we will prove that the integral transformation
\begin{align}\label{E3.50}
 \nonumber \tilde{B}_{\nu,m}\hspace{0.2cm}:L^{2}(\mathbb{R},&\hspace{0.2cm} e^{-x^{2}}dx)\longrightarrow \mathcal{A}_{m}^{2,\nu}(\mathbb{C})\\
 &\varphi\longmapsto \tilde{B}_{\nu,m}[\varphi](z):=\int_{-\infty}^{+\infty}\tilde{K}_{\nu,m}(z,x)\varphi(x)e^{-x^{2}}dx
 \end{align}
is an isometry operator. To do so, we prove that the function $\tilde{K}_{\nu,m}(z,.)$ belongs to $L^{2}(\mathbb{R},\hspace{0.2cm}e^{-x^{2}}dx)$
for each fixed $z\in \mathbb{C}.$\\
Applying the Parseval's formula, for the function $\tilde{K}_{\nu,m}(z,.)$ defined in  $(\ref{E3.34})$, we get
\begin{align}\label{E3.51}
 \nonumber \parallel \tilde{K}_{\nu}(z,.)\parallel_{L^{2}(\mathbb{R},\hspace{0.2cm}e^{-x^{2}}dx)}^{2}&
=\sum_{j=0}^{+\infty}\mid\psi_{j}^{m}(z)\mid^{2}\\
\nonumber&=\sum_{j=0}^{m-1}\frac{\nu^{j+1}\mid z^{2}\mid^{j}}{\pi j!}+
\sum_{j=m}^{+\infty}\frac{\nu^{j-m+1}\Gamma(j-m+1)}{\pi(j!)^{2}}\mid z^{2}\mid^{j}\\
&=\sum_{k=0}^{m-1}\frac{\nu\mid \nu z^{2}\mid^{k}}{\pi k!}+
\sum_{k=0}^{+\infty}\frac{\nu^{k+1}\Gamma(k+1)}{\pi[(k+m)!]^{2}}\mid z^{2}\mid^{k+m}.
\end{align}
Then, it is not hard to obtain the following inequality
\begin{align}\label{E3.52}
 \parallel \tilde{K}_{\nu}(z,.)\parallel_{L^{2}(\mathbb{R},\hspace{0.2cm}e^{-x^{2}}dx)}^{2}\leq \frac{\nu}{\pi}(1+\mid z\mid^{2m})\exp(\nu \mid z\mid^{2}).
\end{align}
This proves that $\tilde{K}_{\nu,m}(z,.)\in L^{2}(\mathbb{R},\hspace{0.2cm}e^{-x^{2}}dx)$ for every fixed $z\in \mathbb{C}.$
The integral transform $\tilde{B}_{\nu,m}$ associated with the kernel function $\tilde{K}_{\nu,m}(z,.)$ is defined by
\begin{align}\label{E3.53}
 \nonumber\tilde{B}_{\nu,m}[\varphi](z)&=<\tilde{K}_{\nu,m}(z,.),\overline{\varphi}>_{L^{2}(\mathbb{R},\hspace{0.2cm}e^{-x^{2}}dx)}\\
&=\int_{-\infty}^{+\infty}\tilde{K}_{\nu,m}(z,x)\varphi(x)e^{-x^{2}}dx,
\end{align}
for every $z\in \mathbb{C}$, provided that the integral exists.\\
It is noted that it is quite easy to see that for every nonnegative integer $j\in \mathbb{Z}_{+}$, we have
\begin{align}\label{E3.54}
 \tilde{B}_{\nu,m}[\varphi_{j}](z)=\psi_{j}^{m}(z).
\end{align}
Indeed, this is an immediate consequence of the equations $(\ref{E3.34})$ and $(\ref{E3.53})$.\\
Applying the Cauchy-Schwartz inequality to the first equality in the formula $(\ref{E3.53})$, we get
\begin{align}\label{E3.55}
\mid\tilde{B}_{\nu,m}[\varphi](z)\mid\leq\parallel\tilde{K}_{\nu,m}(z,.)\parallel_{L^{2}(\mathbb{R},\hspace{0.2cm}e^{-x^{2}}dx)}
\parallel\varphi\parallel_{L^{2}(\mathbb{R},\hspace{0.2cm}e^{-x^{2}}dx)}.
\end{align}
From the inequality given $(\ref{E3.52})$, one can obtain the following estimate
\begin{align}\label{E3.56}
\mid\tilde{B}_{\nu,m}[\varphi](z)\mid\leq \sqrt{\frac{\nu}{\pi}[(1+\mid z\mid^{2m})\exp(\nu \mid z\mid^{2})]}
\parallel\varphi\parallel_{L^{2}(\mathbb{R},\hspace{0.2cm}e^{-x^{2}}dx)},
\end{align}
for all $\varphi \in L^{2}(\mathbb{R},\hspace{0.2cm}e^{-x^{2}}dx)$ and $z$ fixed in $\mathbb{C}$. The inequality $(\ref{E3.56})$
 ensures the continuity of the following linear functional
\begin{align}\label{E3.57}
\nonumber L^{2}(\mathbb{R},\hspace{0.2cm} &e^{-x^{2}}dx)\longrightarrow \mathbb{C}\\
 &\varphi\longmapsto \tilde{B}_{\nu,m}[\varphi](z).
 \end{align}
The continuity of the above linear functional ensure that we have
\begin{align}\label{E3.58}
 \nonumber \tilde{B}_{\nu,m}[\varphi](z)&=\sum_{j=0}^{+\infty}\lambda_{j}\tilde{B}_{\nu,m}[\varphi_{j}](z)\\
&=\sum_{j=0}^{+\infty}\lambda_{j}\psi_{j}^{m}(z),
\end{align}
for each fixed $z\in \mathbb{C}$ and every $\varphi=\displaystyle{\sum_{j=0}^{+\infty}}\lambda_{j}\varphi_{j}$ in $L^{2}(\mathbb{R},\hspace{0.2cm}e^{-x^{2}}dx)$.
 Moreover, $\tilde{B}_{\nu,m}[\varphi](z)$ converges absolutely for all $z\in \mathbb{C}$. Indeed, using the Cauchy-Schwarz inequality and $(\ref{E3.52})$,
  we obtain
\begin{align}\label{E3.59}
 \nonumber \mid \tilde{B}_{\nu,m}[\varphi](z)\mid&\leq \sum_{j=0}^{+\infty}\mid\lambda_{j}\mid \mid \psi_{j}^{m}(z)\mid\\
\nonumber &\leq  (\sum_{j=0}^{+\infty}\mid\lambda_{j}\mid^{2})^{\frac{1}{2}}(\sum_{j=0}^{+\infty}\mid \psi_{j}^{m}(z)\mid)^{\frac{1}{2}}\\
&\leq \sqrt{\frac{\nu}{\pi}[(1+\mid z\mid^{2m})\exp(\nu \mid z\mid^{2})]}\parallel\varphi\parallel_{L^{2}(\mathbb{R},\hspace{0.2cm}e^{-x^{2}}dx)}.
 \end{align}
Furthermore, we have
\begin{align}\label{E3.60}
 \nonumber \parallel \tilde{B}_{\nu,m}[\varphi]\parallel_{\nu,m}^{2}&=\sum_{j=0}^{+\infty} \mid \lambda_{j}\mid^{2}\parallel \psi_{j}^{m}\parallel_{\nu,m}^{2}\\
 \nonumber &=\sum_{j=0}^{+\infty} \mid \lambda_{j}\mid^{2}\\
&=\parallel\varphi\parallel^{2}_{L^{2}(\mathbb{R},\hspace{0.2cm}e^{-x^{2}}dx)}.
 \end{align}
This shows that $\tilde{B}_{\nu,m}$ is a well defined isometry from $L^{2}(\mathbb{R},\hspace{0.2cm}e^{-x^{2}}dx)$ onto $\mathcal{A}_{m}^{2,\nu}(\mathbb{C})$.
Finally, by considering the following isometry
\begin{align}\label{E3.61}
 \nonumber B_{\nu,m}\hspace{0.2cm}:L^{2}(\mathbb{R},&\hspace{0.2cm} dx)\longrightarrow \mathcal{A}_{m}^{2,\nu}(\mathbb{C})\\
 &\varphi\longmapsto B_{\nu,m}[\varphi]=\tilde{B}_{\nu,m}\circ T[\varphi],
 \end{align}
where $T$ is the canonical isometry given by
\begin{align}\label{E3.62}
 \nonumber T\hspace{0.2cm}:L^{2}(\mathbb{R},&\hspace{0.2cm} dx)\longrightarrow L^{2}(\mathbb{R},\hspace{0.2cm} e^{-x^{2}}dx)\\
 &\varphi\longmapsto e^{\frac{x^{2}}{2}}\varphi,
 \end{align}
we give the desired result and then the proof of theorem
$(\ref{T3.1})$ is completed.
\end{proof}
\begin{proof}[\textbf{Proof of the proposition }$(\ref{P3.1})$]
After a direct application of the theorem $(\ref{T3.1})$, it remains just to give an explicit formula for the function $\varpi_{m}(t)$,
defined in $(\ref{E3.20})$, in the case when $m=2$. For this, we proceed as follows. Explicitly, we have
\begin{align}\label{E3.63}
\nonumber\varpi_{2}(t)&=(\sqrt{t}e^{-t})*(\sqrt{t}e^{-2t})\\
\nonumber&=\int_{0}^{t}\sqrt{s}e^{-s}\sqrt{t-s}e^{-2(t-s)}ds\\
\nonumber&=e^{-2t}\int_{0}^{t}\sqrt{t-s}\sqrt{s}e^{s}ds\\
\nonumber&=e^{-2t}\int_{0}^{1}\sqrt{t-tx}\sqrt{tx}e^{xt}tdx,\hspace{0.25cm} (x:=\frac{s}{t})\\
&=t^{2}e^{-2t}\int_{0}^{1}\sqrt{1-x}\sqrt{x}e^{xt}dx.
\end{align}
Using the following integral representation of the hypergeometric function $\prescript{}{1}{F}_1^{}(\alpha,\beta;t)$ (\cite[p.331]{Nik})
\begin{align}\label{E3.64}
\prescript{}{1}{F}_1^{}(\alpha,\beta;t)=\frac{\Gamma(\beta)}{\Gamma(\alpha)\Gamma(\beta-\alpha)}\int_{0}^{1}(1-x)^{\beta-\alpha-1}x^{\alpha-1}e^{tx}dx,
 \hspace{0.25cm} 0<\alpha<\beta,
\end{align}
for $\alpha=\frac{3}{2}$ and $\beta=3$, the equation $(\ref{E3.63})$ can be rewritten as
\begin{align}\label{E3.65}
\nonumber \varpi_{2}(t)&=\frac{\Gamma(\frac{3}{2})\Gamma(\frac{3}{2})}{2}t^{2}e^{-2t}\prescript{}{1}{F}_1^{}(\frac{3}{2};3;t)\\
&=\frac{\pi}{8}t^{2}e^{-2t}\prescript{}{1}{F}_1^{}(\frac{3}{2};3;t),
\end{align}
where we have used $\Gamma(\frac{3}{2})=\frac{\sqrt{\pi}}{2}$ and $\Gamma(3)=2!=2$.\\
The proof of proposition is closed.
\end{proof}
\section{Conclusion and open questions}
In this  paper,  we have  reconsidered the  study  of  the
Bargmann-Dirichlet space on  the complex plane $\mathbb{C}$ and its
generalizations considered in  \cite{Elh}. In particular, we    have
given   a new characterization of such spaces as harmonic spaces of
the magnetic Laplacian with suitable domains. Then, we  have
elaborated the associated unitary
integral transforms  of  Bargmann type.\\
The present work comes up with many open  questions. In particular,
it will be intersecting to make contact with quantum dynamic
activities on such   non trivial spaces in the presence of magnetic
sources. Precisely, it would be of interest to bring a physical
interpretation of states  belonging to the studied spaces. Such
questions could be addressed elsewhere.

\end{document}